\documentclass[aps,pra,reprint,groupedaddress]{revtex4-1}
\usepackage{xcolor}
\usepackage{hyperref}
\hypersetup{colorlinks=false,linkbordercolor=red,linkcolor=green,pdfborderstyle={/S/U/W 1}}
\usepackage[UKenglish]{babel}
\usepackage[]{graphicx, amsfonts, amsmath, amssymb, amstext, latexsym, float, color, hyperref, mathtools}

\newcommand{\be}{\begin{equation}}
\newcommand{\ee}{\end{equation}}

\newcommand{\ket}[1]{|#1\rangle}

\newcommand{\eat}[1]{}
\newcommand{\E}{\mathcal{E}}

\newtheorem{theorem}{Theorem}

\newtheorem{conjecture}[theorem]{Conjecture}

\newtheorem{lemma}[theorem]{Lemma}

\newenvironment{proof}[1][Proof]{\noindent\textbf{#1.} }{\ \hfill\rule{0.5em}{0.5em}}

\begin{document}

\title{On capacity of optical communications over a lossy bosonic channel with a receiver employing the most general coherent electro-optic feedback control}

\author{Hye Won Chung}
\thanks{Email of corresponding author: hyechung@umich.edu. This paper was presented in part at the 2011 IEEE International Symposium on Information Theory (ISIT)~\cite{chung2011capacityISIT} and the IEEE 49th Annual Allerton Conference on Communication, Control, and Computing (Allerton)~\cite{chung2011capacity}.}
\author{Saikat Guha$^{\dagger}$}
\author{Lizhong Zheng$^\circ$}
\affiliation{$^*$EECS Department, University of Michigan, 1301 Beal Avenue, Ann Arbor, MI USA 48109\\
$^\dagger$Quantum Information Processing group, Raytheon BBN Technologies, 10 Moulton Street, Cambridge, MA USA 02138 \\
$^\circ$ EECS Department, MIT, 77 Massachusetts Avenue, Cambridge, MA USA 02139}


\begin{abstract}
We study the problem of designing optical receivers to discriminate between multiple coherent states using {\em coherent processing} receivers---i.e., one that uses arbitrary coherent feedback control and quantum-noise-limited direct detection---which was shown by Dolinar to achieve the minimum error probability in discriminating any two coherent states. We first derive and re-interpret Dolinar's binary-hypothesis minimum-probability-of-error receiver as the one that optimizes the information efficiency at each time instant, based on recursive Bayesian updates within the receiver. Using this viewpoint, we propose a natural generalization of Dolinar's receiver design to discriminate $M$ coherent states each of which could now be a codeword, i.e., a sequence of $N$ coherent states each drawn from a modulation alphabet. We analyze the channel capacity of the pure-loss optical channel with a general coherent-processing receiver in the low-photon number regime and compare it with the capacity achievable with direct detection and the Holevo limit (achieving the latter would require a quantum joint-detection receiver). We show compelling evidence that despite the optimal performance of Dolinar's receiver for the binary coherent-state hypothesis test (either in error probability or mutual information), the asymptotic communication rate achievable by such a coherent-processing receiver is only as good as direct detection. This suggests that in the infinitely-long codeword limit, all potential benefits of coherent processing at the receiver can be obtained by designing a good code and direct detection, with no feedback within the receiver.
\end{abstract}


\maketitle

Over time $t\in [0, T)$, consider a {\it coherent-state} input of constant amplitude ${S}$ to a pure-loss optical channel, where $S \in \mathbb{C}$, and $|S|^2T$ is the mean photon number. Coherent state is the quantum description of light generated by an ideal laser. In a noise-free environment, if one uses an ideal quantum-noise-limited photon counter to receive this optical signal, the output of the photon counter is a Poisson point process, with rate $\lambda=|S|^2$ over the time period $[0, T)$, indicating arrivals of individual photons. Clearly, one can generalize from a constant input to an arbitrary temporal-mode shape of the coherent-state pulse $S(t)$, $t\in [0, T)$, which if detected with an ideal photon counter would result in a non-homogeneous Poisson process of rate $\lambda(t) = |S(t)|^2$. The mean number of photons, $\int_0^T |S(t)|^2 dt$, expended in the transmitted pulse, is the natural metric quantifying communication cost. A photon counter with sub-unity detection efficiency $\eta \in (0, 1]$ can be modeled as a lossy channel of transmissivity $\eta$ followed by ideal photon counting. Further, a coherent state at the input of a lossy channel appears as a coherent state at the output of the channel with its amplitude scaled by the channel's transmissivity $\eta \in (0, 1]$. Therefore, without loss of generality, in this paper we will assume a lossless channel and unity-efficiency photodetection, with an implicit scaling of any constraint imposed on the transmitted mean photon number per mode for all the communication-rate calculations. Receivers that are based on counting photons, i.e., detecting the intensity of the optical signals, are called direct-detection receivers, and the resulting communication channel when coherent states are used for input modulation, is called a Poisson channel. The capacity of the Poisson channel has been well studied~\cite{Shamai90,Wyner88_1,wang2014refined}.

Since a coherent-state optical signal can be described by a complex amplitude $S$, it is of interest to design coherent receivers that measure the phase of $S$, and thus allow information to be modulated on the phase. The standard optical receivers that can detect the phase of the input coherent state are homodyne and heterodyne detection receivers, which mix the received coherent state with a strong coherent-state local oscillator (at the same carrier frequency as the input for homodyne, and at a slight carrier-frequency offset for heterodyne) on a 50-50 beamsplitter and detect the two outputs of the beamsplitter by a pair of linear-mode photodetectors followed by integrating the difference of their output photocurrents. However, we will consider the following lesser-known receiver architecture to detect the phase of an optical signal, proposed by Kennedy (see Figure \ref{fig:dolinar}). 

Instead of directly feeding the input coherent state of complex amplitude $S$ into the photon counter, Kennedy's receiver mixes the input signal with a fixed-amplitude strong coherent-state local oscillator of amplitude $l/\sqrt{1-\gamma}$ on a highly transmissive beamsplitter (of transmissivity $\gamma \approx 1$), and detects the output of the beamsplitter, which is a coherent state of amplitude $S+l$, with an ideal photon detector. The output of the photon counter therefore is a Poisson process with rate $|S+l|^2$. In principle, $l$ can be chosen as an arbitrary complex number, with any desired phase difference from the input signal $S$. Thus, the output of this processing can be used to extract phase information in the input. In a sense, the local control signal is designed to control the channel through which the optical signal $S$ is observed. 

Kennedy used this architecture to distinguish between binary coherent-state hypotheses, i.e., two candidate coherent-state temporal waveforms $S_0(t), S_1(t), t\in [0,T)$, with prior probabilities $\pi_0, \pi_1$, respectively, using a control signal whose complex amplitude $l$ was held constant in $[0,T)$. This was later generalized by Dolinar \cite{DOL73}, who used a time-varying control waveform $l(t), t\in [0,T)$, which flip-flopped between two pre-determined waveforms $l_0(t)$ and $l_1(t)$ adaptively at each photon arrival instant at the detector. Dolinar showed that the local signal waveforms $l_0(t)$ and $l_1(t)$ can be designed in a way, such that the resulting average probability of error for the aforesaid binary hypothesis test is given by:

\be
\label{eqn:binary_optimal}
P_e= \frac{1}{2}\left( 1-\sqrt{1-4\pi_0\pi_1 e^{-\int_0^T |S_0(t)-S_1(t)|^2dt}}  \right).
\ee

Rather surprisingly, this error probability {\em exactly} coincides with the minimum average error probability for discriminating the two coherent-state waveforms with {\em any} measurement allowed by quantum mechanics, which we will refer to as the Yuen-Kennedy-Lax (YKL) limit \cite{YKL1975, Hel76}. The optimality of Dolinar's receiver is an amazing result, as it shows that the minimum-probability-of-error quantum measurement for the binary coherent-state hypothesis test problem can be implemented with the very simple receiver structure shown in Figure \ref{fig:dolinar}, whose functioning can be described completely in terms of semi-classical (shot-noise) theory of photo detection. Unfortunately, this does not generalize to problems involving discrimination of more than two coherent states, where it appears that the receiver must employ truly non-classical effects in order to exactly attain the YKL limit~\cite{Silva2013}.

The goal of this paper is twofold. We are interested in finding a natural generalization of Dolinar's receiver to general hypothesis testing problems with more than two possible signals. In addition, we also consider using such receivers to receive coded transmissions, and thus compute the asymptotic information rate that can be reliably carried through the optical channel. Our investigation will be specifically tied to structure of the receiver front-end shown in Figure \ref{fig:dolinar}, where we will allow the control signal to be varied arbitrarily over the entire received modulated codeword. In Section~\ref{sec:binary}, we will begin by re-deriving Dolinar's design of the optimal control waveform $l(t)$ for the binary case using a method different from Dolinar's, in order to motivate our more general approach. In Section~\ref{sec:gen_M}, we will discuss the performance of the Dolinar receiver front end to discriminate $M > 2$ coherent states, when the time-incremental optimization of a class of R{\' e}nyi information metrics is used to design the local control signal. In Section~\ref{sec:coded}, we consider the performance of this receiver for optimizing the asymptotic information communication rate, and prove the following no-go theorem. The Kennedy-Dolinar receiver acting directly on the received codeword, where the control signal is kept constant over each modulation symbol but is allowed to vary across the $N$ symbols in a codeword, can perform no better than an direct-detection receiver with no internal feedback, in the limit of $N \to \infty$. We conjecture that even if we were to allow the coherent-state codeword to be processed by an arbitrary passive linear-optical mixer prior to feeding it into the Dolinar receiver, and the control signal to be varied arbitrarily over the entire time duration of that processed codeword, the result of our no-go theorem would still apply. We however leave open the proof of this fully general result. If this conjectured result were true, it would imply that when the benefit of coding is available, that local coherent feedback within the receiver does not help increase communication rate, thereby suggesting that truly non-classical joint optical processing and detection of the codeword---not describable by the semi-classical theory of photo-detection---would be needed to attain the ultimate (Holevo) limit~\cite{holevo1998quantum} of optical communications capacity. We conclude the paper in Section~\ref{sec:conclusion}.

\begin{figure}
\centering
\includegraphics[scale=0.5]{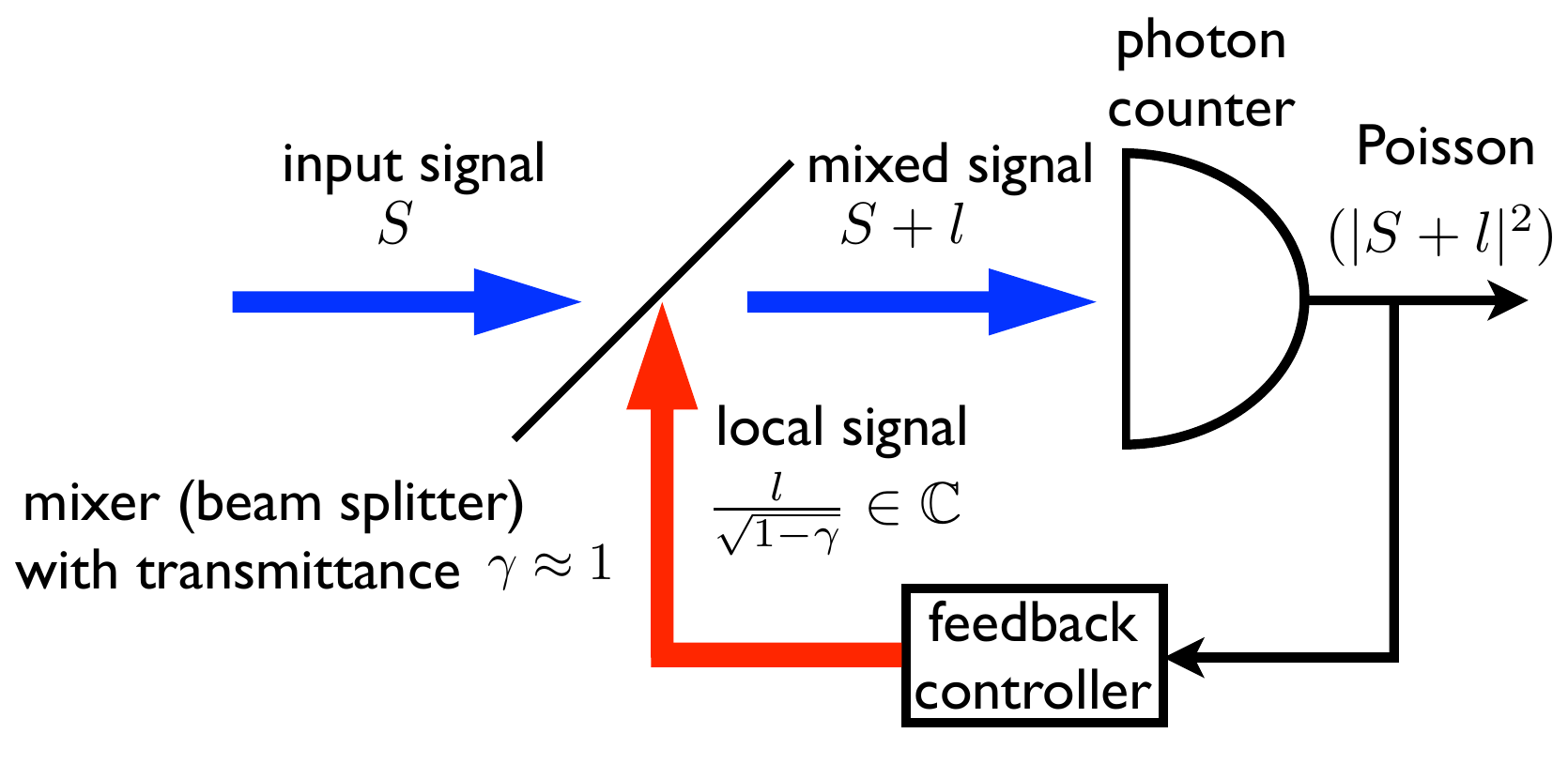}
\caption{Coherent receiver using local feedback signal.}
\label{fig:dolinar}
\end{figure}
\section{Binary Hypothesis Testing}
\label{sec:binary}

Let us consider the binary hypothesis testing problem with two candidate coherent state signals, $\{{S_0(t)}, {S_1(t)}\}$, $t \in [0,T)$ under hypotheses $H=0, 1$, respectively, and denote $\pi_0(t)$ and $\pi_1(t)$ as the posterior distributions over the two hypotheses, conditioned on the output of the photon counter up to time $t$. We assume that $S_0(t), S_1(t) \in \mathbb{R}$. This simplifying assumption accrues no loss of generality for the binary case since we can always choose an axis in the phase space passing through two complex-valued input signals and call that as the `real' axis. 
Based on the receiver's knowledge of the posterior probabilities $\pi_0(t)$ and $\pi_1(t)$ at time $t$, it chooses the control signal $l(t)$ (based on optimizing an incremental information metric to be described shortly) whose value is held constant over the infinitesimal interval $[t, t+\Delta)$. After observing the output of the photon counter during this infinitesimal interval, i.e., based on whether a click appears or not, the receiver updates the posterior probabilities of the hypotheses to obtain $\pi_0(t+\Delta)$ and $\pi_1(t+\Delta)$, and then follows the above procedure again to choose the control signal over the next infinitesimal interval, and so on. 
In the following, we will focus on solving the single step optimization of $l$ (at time $t$) in the above described recursive procedure, and will drop the dependency on $t$ to simplify the notation. 

We first observe that the optimal value of $l$ must be real, as having a non-zero imaginary part in $l$ simply adds a constant rate to the two candidate Poisson point processes (corresponding to the two hypotheses), which cannot improve the quality of observation. When we write $\lambda_i = (S_i +l)^2, i=0,1$ to denote the rate of the resulting Poisson processes, the number of photon arrivals at the output of photon counter during the interval $\Delta$ follows the Poisson distribution
\be
\begin{split}
&\Pr(k \text{ photon arrivals in } \Delta \text{ interval}|H=i)\\
&=\frac{(\lambda_i\Delta)^ke^{-\lambda_i\Delta}}{k!},
\end{split}
\ee
conditioned on which hypothesis ($H = 0, 1$) is true.
Over a very short period of time, i.e., when $\Delta\to 0$, under either hypothesis, the realized Poisson process generates with a high probability either $0$ or $1$ photon arrival, with probabilities $e^{-\lambda_i \Delta}$ and $1-e^{-\lambda_i \Delta}$, respectively~\footnote{One has to be careful in using the binary-output channel as an approximation of the Poisson channel. As we are optimizing over the control signal, it is not obvious that the resulting $\lambda_i$'s are bounded. In other words, the mean of the Poisson distributions, $\lambda_i \Delta$, might not be small. The assumption of either $0$ or $1$ arrival, and the approximation in the corresponding probabilities, can be justified as follows. First, a single photon detector is much more practical, given the current state of technology, that a fully number-resolving high bandwidth photon counter. A single photon detector can sense whether or not any number of photons arrives during a time interval $\Delta$, but cannot count the number of photon arrivals, especially as $\Delta\to0$. So, the binary-output channel model is much more practical than the Poisson-output channel model. Second, when we want to maximize the ability to distinguish between two hypotheses $H=0, 1$, we essentially need to distinguish between the signal amplitudes $S_0$ and $S_1$ using photon arrival events. Adding a feedback control signal $l\to\infty$ does not help in distinguishing $S_0$ and $S_1$. In this sense, we can reason that the optimal $l$ should not make $\lambda_i$ unbounded.}\label{foot:binary}. Over this short period of time, the receiver front end induces a binary-input binary-output channel as shown in Figure \ref{fig:short}, whose parameters depend upon the value of the control signal $l$. Our goal is to pick an $l$ for each short interval such that they contribute to the overall decision in the best possible manner.

The difficulty here is that it is not obvious how we should quantify the contribution of the observation over a short period of time to the performance of the overall decision. Let us consider the intuitive approach where we choose the $l$ that maximizes the mutual information over the induced binary channel at each incremental time step. For convenience, we write the input to the channel as $H\in\{0,1\}$ and the output of the channel as $Y \in \{0,1\}$, indicating either $0$ or $1$ photon arrival. The mutual information between $H$ and $Y$ is given by
\be\label{eqn:IHY1}
I(H;Y)=\sum_{h=0}^1 \pi_h \left(\sum_{y=0}^1 \ln\frac{P_{Y|H}(y|h)}{\left(\sum_{h'=0}^1\pi_{h'}P_{Y|H}(y|h')\right)}\right)
\ee
where $\{\pi_0,\pi_1\}$ are input probabilities and $P_{Y|H}(y|h)$ is the channel distribution. 
The following result gives the solution to this optimization problem of finding the control signal $l^*$ that maximizes $I(H;Y)$.

\begin{lemma}\label{lem:opt_feed}
{\it The optimal choice maximizing the mutual information $I(H;Y)$ in~\eqref{eqn:IHY1} for the effective binary channel is:
\be
\label{eqn:BSCoptimal}
l^* = \frac{S_0 \pi_0 - S_1 \pi_1}{\pi_1-\pi_0}.
\ee
With this choice of the control signal, the following relation holds:
\be
\label{eqn:balance1}
\pi_0 \sqrt{\lambda_0}  =  \pi_1 \sqrt{\lambda_1} .
\ee}
\end{lemma}
\begin{proof}
Appendix \ref{sec:pf_lem_opt_feed}
\end{proof}

\begin{figure}
\centering
\includegraphics[scale=0.9]{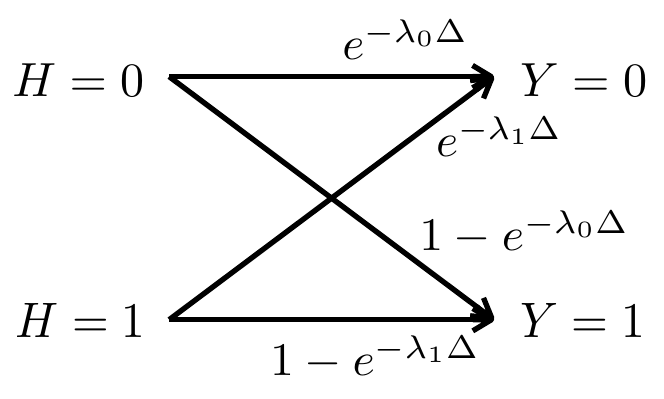}
\caption{Effective binary channel between input hypothesis $H\in\{0,1\}$ and output of the photon counter $Y\in\{0,1\}$, indicating either 0 or 1 photon arrival over an infinitesimal time interval of length $\Delta$.}
\label{fig:short}
\end{figure}

The relation in~\eqref{eqn:balance1} lends some useful insights. If $\pi_0 > \pi_1$, we have $\lambda_1 > \lambda_0$, and vice versa. That is, by switching the sign of the control signal $l$, we always make the Poisson rate corresponding to the hypothesis with the higher probability smaller. In the short interval where this control is applied, with a high probability we would observe no photon arrival, in which case we would confirm the more likely hypothesis. For a very small value of $\Delta$, this occurs with a dominating probability, such that the posterior distribution changes only by a very small amount.
On the other hand, when there is a photon arrival, i.e., $Y=1$, we would be quite surprised, and the posterior distribution of the hypotheses moves away significantly from the prior. Considering this latter case, the updated distribution over the hypotheses can be written as:
\be
\frac{\Pr(H=1|Y=1)}{\Pr(H=0|Y=1)} = \frac{\pi_1 \cdot \lambda_1 \Delta}{\pi_0 \cdot \lambda_0 \Delta} +O(\Delta)= \frac{\pi_0}{\pi_1}+O(\Delta).
\ee
The posterior distributions under $0$ or $1$ photon arrival turn out to be inverse of one another in the $\Delta\to0$ limit. In other words, the larger one of the two probabilities $\pi_0(t)$ and $\pi_1(t)$ remains the same no matter if there is an arrival in the interval or not. As we apply such optimal control signals recursively, this larger value smoothly progresses towards $1$ at a predictable rate in $t \in [0,T)$, regardless of when and how many photon arrivals were actually observed. In other words, {\it the random photon arrivals only affect the decision on which is the more likely hypothesis, but do not affect the quality of this decision.} The following lemma describes this recursive control signal and the resulting receiver performance. Without loss of generality, we assume that at $t=0$, the prior distribution satisfies $\pi_0 \geq \pi_1$. Also we let $N(t)$ denote the number of photon arrivals observed in the interval $[0,t)$.

\begin{lemma}
{\it \label{lemma:binaryoptimality}
Let $g(t)$ satisfy $g(0)=\pi_0/\pi_1$ and
\be
g(t) = g(0) \, \exp\left[ \int_0^t \frac{(S_0(t)-S_1(t))^2 (g(\tau)+1)}{g(\tau)-1} d\tau\right].
\ee
The recursive mutual-information-maximization procedure described above yields a control signal 
\be
l^*(t) = \left\{ 
\begin{array}{ll}
l_0(t) & \quad \mbox{if } N(t) \mbox{ is even}\\
l_1(t) & \quad \mbox{if } N(t) \mbox{ is odd}
\end{array}
\right.
\ee
where,
\be
l_0(t) = \frac{S_1(t)- S_0(t) g(t)}{g(t)-1}, \quad
l_1(t) = \frac{S_0(t)- S_1(t) g(t)}{g(t)-1}.
\ee
Furthermore, at time $T$, the decision of the hypothesis testing problem is $\widehat{H}=0$ if $N(T)$ is even, and $\widehat{H}=1$ otherwise. The resulting probability of error coincides with~\eqref{eqn:binary_optimal}. }
\end{lemma}
\begin{proof}
Appendix \ref{sec:pf_lem_opt_feed1}
\end{proof}

Figure \ref{fig:samplepath} shows an example of the optimal control signal. The plot is for a case where $S_i(t)$'s are constant on-off-keying waveforms; i.e., $S_0(t) = 0$ and $S_1(t) = S$ $\forall t \in [0, T)$. As shown in the plot, the control signal $l(t)$ jumps between two prescribed curves, $l_0(t), l_1(t)$, corresponding to the cases $\pi_0(t)> \pi_1(t)$ and $\pi_0(t)< \pi_1(t)$, respectively. With the optimal choice of the control signal, at each instant of a photon arrival, the receiver is maximally surprised and it flips its choice of the hypothesis $\widehat{H}$. However, $g(t) = \max\{\pi_0(t), \pi_1(t)\}/\min\{\pi_0(t), \pi_1(t)\}$, indicating how much the receiver is committed to the more likely hypothesis, increases at a steady rate regardless of the actual arrival events. 

\begin{figure}
\centering
\includegraphics[scale=0.65]{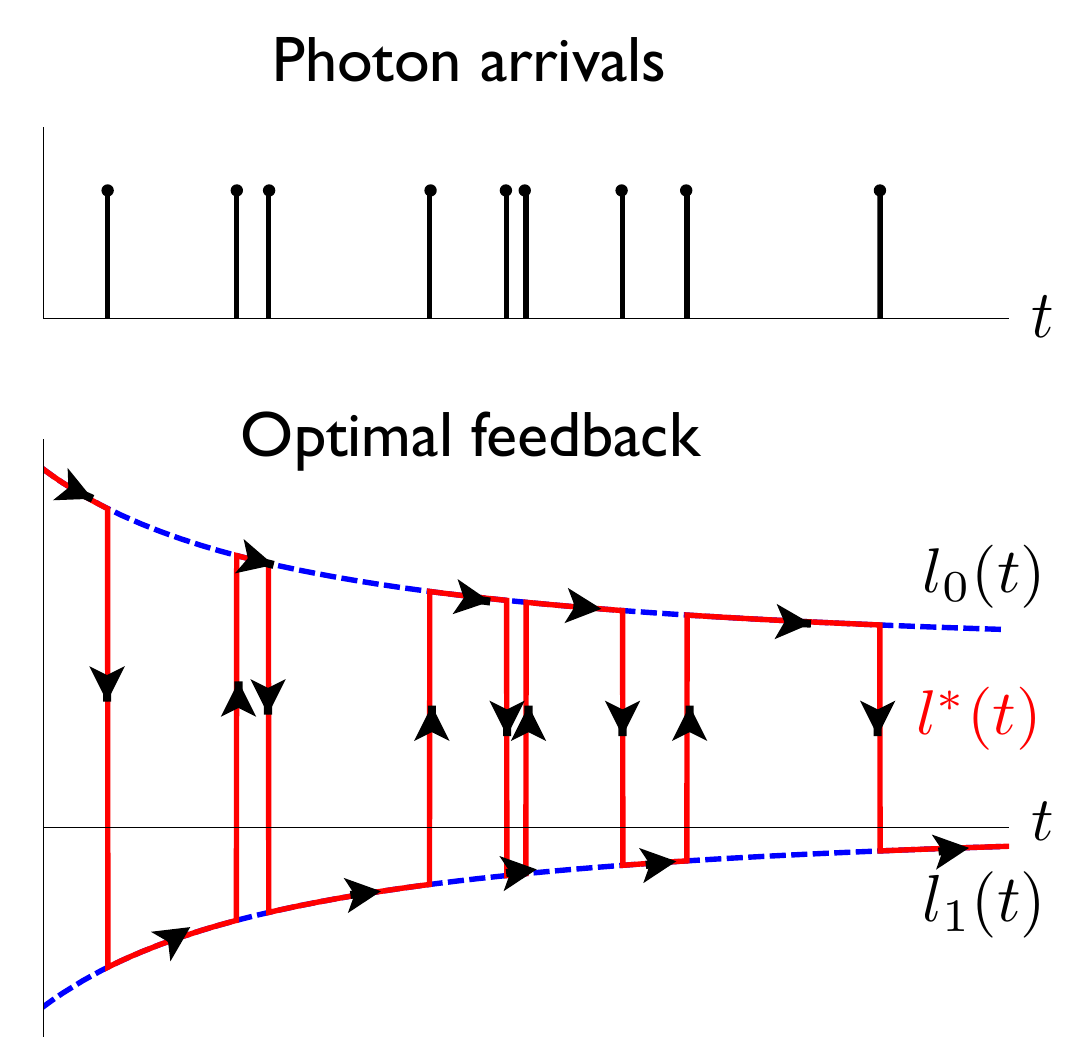}
\caption{An example of the control signal $l^*(t)$, which jumps between two pre-determined waveforms $l_0(t)$ and $l_1(t)$ adaptively at each photon arrival instant at the detector. This control signal achieves the minimum probability of error for binary hypothesis testing for discriminating on-off keying coherent-state signals.}
\label{fig:samplepath}
\end{figure}

Before we go on to the more general $M$-ary setting, a few comments are in order. Takeoka generalized Dolinar's original result---which was derived specifically for optimally discriminating between two coherent states---to show that the receiver front end shown in Figure \ref{fig:dolinar} can actually realize an arbitrary binary projective measurement on an arbitrary set of (one of two) input states~\cite{Takeoka2007}. Takeoka posed the problem of minimum-error discrimination of two non-orthogonal states as the (zero-error) discrimination of two mutually orthogonal states that correspond to the YKL measurement projectors. He chose the control signals in such a way that if the receiver is fed with one of these two orthogonal states, that at every incremental time step in $[0, T)$, the conditional states under the two hypotheses remain orthogonal. Takeoka's construction proved a special case of an earlier result by Walgate {\em et al.}~\cite{Walgate2000} which states that when presented with many copies of one of two pure states, there always exists a sequence of projective measurements that act on each copy individually while feeding forward the measurement result towards determing the measurement to be performed on the next copy---also termed {\em local operations and classical communications} (LOCC)---which can attain the quantum minimum error probability in choosing between the two hypotheses, and in turn also satisfy the aforesaid condition of incremental orthogonality of YKL projectors as one progresses through the copies that Takeoka's construction guarantees. Given Walgate {\em et al.}'s result on an LOCC strategy being always optimal for binary multi-copy pure state discrimination, the fact that Dolinar's receiver exactly attains the YKL limit is not so surprising in hindsight. In the same paper~\cite{Walgate2000}, Walgate {\em et al.} argue that for $M$-ary hypothesis testing, an LOCC strategy is not always globally optimal. Even though this does not imply that the Kennedy-Dolinar receiver front end will {\em not} attain the YKL limit of $M$-ary coherent state discrimination, it is highly indicative of that being so. 

Finally, it is well known that for an ensemble of $M=2$ pure states, the measurement that minimizes the error probability (i.e., attains the YKL conditions) is the same as the measurement that maximizes the mutual information (or, accessible information), and is a $2$-output projective measurement. Hence, it is not surprising that our derivation of the control signal $l^*(t)$, which was based on maximizing the incremental mutual information, results in the same answer as what Dolinar derived. It is worth noting however that for $M>2$ pure states, the YKL measurement---which is an $M$-output projective measurement---is in general {\em different} from the one that maximizes the accessible information, which in general is $d$-output measurement described by {\em positive operator valued measure} (POVM) operators with $M \le d \le M(M+1)/2$.

\section{Generalization to $M$-ary Hypothesis Testing}\label{sec:gen_M}

Our success in interpreting the binary hypothesis testing problem as an incremental maximization of mutual information gives us useful insights on designing a general communications receiver. Regardless of the physical channel that one communicates over, one can always contemplate designing a receiver that builds up a ``slow motion" understanding of the received signal by studying how the posterior distribution over the messages evolves over time (during the demodulation and decoding of the modulated message). This evolving posterior distribution, conditioned on more and more observations at the receiver, would be expected to drive the uniform prior towards an eventual deterministic distribution, thus allowing the receiver to ``lock in" on a particular message. This viewpoint is more general than the conventional setup in information theory, and is particularly useful in understanding dynamic problems, as it is not based on any notion of sufficient statistics, block codes, or any predefined notions of reliability. As we measure how far the posterior distribution moves at each time instant, we can quantify how the communication transmission and reception process at each time instant contributes to the overall decision making.

The optimality result in Lemma \ref{lemma:binaryoptimality} is, however, difficult to duplicate for general $M$-ary problems. We can of course always mimic the procedure, i.e., choose the control signal that maximizes the incremental mutual information over an $M$-input-binary-output channel at each time instant (binary output corresponding to no photon arrival and one photon arrival in the incremental interval). However, we have found that the resulting control signal does not always give the minimum probability of error. The reason for this is intuitive. There is a fundamental difference between maximizing mutual information and minimizing the probability of error for an ensemble of size $M>2$. A posterior distribution with a lower entropy does not necessarily correspond to a lower probability of error in discriminating the states in the ensemble. These two coincide only for the binary case, since the posterior distribution over two messages lives in a single-dimensional space. In general, the goal of decision making favors the posterior distribution that has a dominating largest element, whereas maximizing mutual information does not impose such a requirement on the posterior and is agnostic to the exact form of the posterior as long as `information' conveyed is maximized.

Consequently, it is hard to define a metric on the efficiency of communication over a small time interval in the middle of a communication session that can precisely measure how well the measurement performed in the interval serves the overall purpose (of choosing between the encoded-modulated messages at a minimum probability of error, for instance). Even if one could define such a metric, it is conceivable that an analytical solution of the optimal control signal by a time-incremental optimization of that metric might be hard. Such an incremental metric, if one exists, should be time-varying, i.e., should be able to adapt itself based upon how much time is left before the decision must be finalized. Intuitively, at an early instant in time (i.e., when a longer time remains before the final decision needs to be made), since the current observation is yet to be combined with many more future observations, the receiver should be more keen to take risk and extract any kind of `information' that is available, and hence it makes sense to maximize mutual information. On the other hand, as the decision deadline approaches, the receiver ought to become progressively more picky in choosing what information to extract from subsequent measurements, and demand only information that helps the receiver lock in to one particular message. Thus, the control signal should be optimized accordingly over the entire duration of receiving the modulated message. 

To test this intuition, we restrict our attention to the family of R\'{e}nyi entropy. R\'{e}nyi entropy of order $\alpha$ of a given distribution $P$ over an alphabet ${\cal X}$ is defined as 
\be
H_\alpha(P) = \frac{1}{1-\alpha} \log \left( \sum_{x \in {\cal X}}P^\alpha(x) \right).
\ee
It is easy to verify that as $\alpha\to 1$, $H_\alpha(P)$ is the Shannon entropy, and as $\alpha\to\infty$,
$H_\infty(P) = -\log \left( \max_{x\in {\cal X}} P(x)\right)$,
which is a measure of the probability of error in guessing $X$, with distribution $P$, since $\hat{X}= \arg\max_x P(x)$. 

Now for general $M$-ary hypothesis testing problems, we consider a recursive design of the control signal $l$ similar to that introduced in Section \ref{sec:binary}, except that at each time instant, rather than maximizing the mutual information over the effective channel, which is equivalent to minimizing the conditional Shannon entropy of the messages, we instead minimize the average R\'{e}nyi-$\alpha$ entropy, i.e., we solve the optimization problem:
\be
\label{eqn:opt_Renyi}
\min_l \sum_{y} P_Y(y) \cdot H_\alpha(P_{H|Y=y}(\cdot)).
\ee

Intuitively, for $\alpha \in [1, \infty)$, as $\alpha$ grows larger, the optimization in~\eqref{eqn:opt_Renyi} tends more in favor of posterior distributions that are concentrated on a single entry. Smaller values of $\alpha$, on the other hand, correspond to being more agnostic to what type of information is obtained as long as the quantity of information being obtained is maximized. A good design should use smaller values of $\alpha$ at the beginning of the communication session and increase $\alpha$ as the decision deadline approaches. We show a numerical example in Figure \ref{fig:numerical} to illustrate this point. We consider discriminating $M=3$ coherent states each with a constant real amplitude, and compare the following two cases: one in which $\alpha = 1$ is held fixed throughout $t \in [0,T]$ and another in which $\alpha = 100$ is held fixed in $t \in [0,T]$. Our intuition says that choosing a smaller $\alpha$ is desirable, when we have enough time to collect information before the final decision. On the other hand, when we need to make a final which-message decision immediately, a larger $\alpha$ is preferable. We observe that using $\alpha=1$ yields better error-probability performance if $T$ is longer, whereas $\alpha = 100$ yields a lower error probability when $T$ is small. 

It will be interesting in future work, to examine the error-probability performance of the Kennedy-Dolinar receiver front end with a control signal designed by using the above incremental R{\' e}nyi-information optimizing technique with an optimal $\alpha(t)$. Moreover, it will be interesting to investigate utilizing a non-coherent-state control signal, for instance a squeezed state.  

\begin{figure}
\centering
\includegraphics[width=\columnwidth]{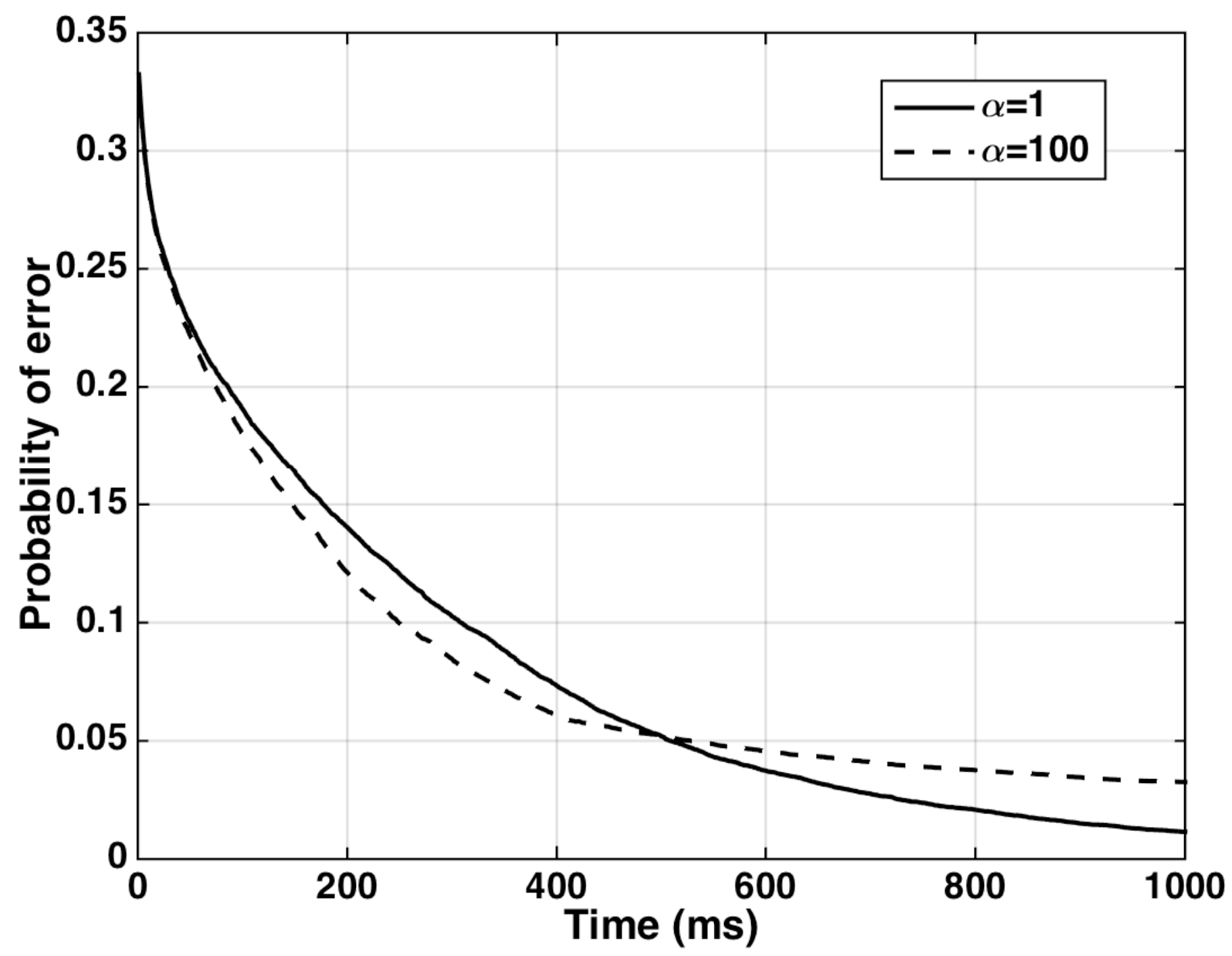}
\caption{Empirical average of detection error probability (after 10,000 Monte Carlo simulations) for ternary hypothesis testing, using control signals that minimize the average R\'{e}nyi $\alpha$-entropy for different values of $\alpha$; Ternary inputs $\{\ket{5},\ket{-6},\ket{3}\}$ are used with prior probabilities $p=\{0.8,0.1,0.1\}$.}
\label{fig:numerical}
\end{figure}

\section{Coded Transmissions and Capacity Results}\label{sec:coded}

Even though the discussion in Section~\ref{sec:gen_M} and the numerical example therein with three coherent states gave us useful insight on optimizing the control signal for hypothesis testing problem, intuition from channel coding tells us that this optimization is a more pertinent question when exponentially many ($M = e^{NR}$) messages are each encoded into a sequence of $N$ coherent states, forming a codebook. Coding-theory intuition further tells us that those $M$ coherent-state sequences, for a good code, should get close to perfectly distinguishable as the codeword length $N$ becomes long, if the rate of the code $R$ is smaller than the capacity $C$, where $C$ is a function of the channel induced by the choice of the optical receiver. In this section, we study the capacity of an optical channel with the Kennedy-Dolinar receiver acting directly on the received codeword, where a control feedback signal in the receiver is chosen to maximize the information rate of the induced channel.


The transmission of an ideal laser-light pulse over a lossy optical channel can be modeled as a pure-state classical quantum channel ${\cal N}_{\eta}: {S} \to |\sqrt{\eta}{S}\rangle$, where ${S} \in {\mathbb C}$ is the complex field amplitude (of the coherent state $|{S}\rangle$) at the input of the channel, $\eta \in (0,1]$ is the transmissivity (the fraction of input power that appears at the output), and $|\sqrt{\eta}{S}\rangle$ is a coherent state at the channel's output.
We are interested in attaining the classical capacity of this channel, i.e., the number of information bits that can be modulated into the optical signals, and reliably decoded with the receiver architecture shown in Figure \ref{fig:dolinar}. Since a coherent state $|S\rangle$ of mean photon number $\E = |{S}|^2$ transforms into another coherent state $|\sqrt{\eta}\,{S}\rangle$ of mean photon number $\eta\E$ over the lossy channel, we will henceforth, without loss of generality, subsume the channel loss in the energy constraint, and pretend that we have a lossless channel $(\eta = 1)$ with a mean-photon-number constraint $\mathbb{E} [|S|^2]\leq \E$ per mode (or per `channel use').

We consider the case where the average number $\E$ of transmitted photons per mode is small, and hence a high photon information efficiency, in bits/photon, is achievable. We are particular interested in analyzing the gap between the capacity with the Kennedy-Dolinar receiver and the Holevo limit, the ultimate achievable capacity with any joint quantum measurement.
At high transmit powers, it is well-known that the Shannon capacity associated with heterodyne detection is close to the Holevo limit. 
In the analysis of the capacity under the mean-photon-number constraint, we will use $o(\cdot)$ and $O(\cdot)$ notations to describe the behavior of functions of the mean photon number $\E$ in the regime of $\E\to 0$. A function described as $o(f(\E))$ and that described as $O(f(\E))$  satisfies
\be
\lim_{\E\to 0}\Bigg|\frac{o(f(\E))}{f(\E)}\Bigg|=0,\quad \limsup_{\E\to 0}\Bigg|\frac{O(f(\E))}{f(\E)}\Bigg|<\infty,
\ee
respectively.

The capacity of the pure-loss ($\eta=1$) optical channel without the constraint in the receiver architecture is studied in \cite{Holevo98, Schumacher97}. It is shown \cite{giovannetti2004classical} that the capacity of the channel (in nats per channel use) is given by 
\be
\label{eqn:holevo}
C_{\sf Holevo} (\E) = (1+\E) \log (1+\E) - \E \log \E,
\ee
where $\E$ is the average number of photons transmitted per channel use. To achieve this data rate, an optimal joint quantum measurement over a long sequence of symbols must be used. In practice, however, such measurement is very hard to implement. We are therefore interested in finding the achievable data rate when a simple receiver structure is adopted. Nevertheless,~\eqref{eqn:holevo} serves as a performance benchmark. In our regime of interest, i.e., $\E\to 0$, it is useful to approximate~\eqref{eqn:holevo} as 
\be
\label{eqn:holevo_approximation}
C_{\sf Holevo}(\E) = \E \log \frac{1}{\E} + \E +o(\E).
\ee

As another performance benchmark, let us consider the Shannon capacity of the channel induced by an ideal direct-detection receiver (no local oscillator mixing or feedback). The capacity of this channel---the Poisson channel---was studied in \cite{Wyner88_1,Shamai90}, and the regime of low average photon numbers was studied in \cite{Ligong08}. For our purposes of performance comparison, we need a more precise scaling law of rate performance, which the following lemma states.

\begin{lemma}[Capacity of Direct Detection]\label{lem:cap_DD}
{\it As $\E\to 0$, the optimal input distribution to the optical channel with a direct-detection receiver is on-off-keying, with 
\be
\ket{S} = \left\{ \begin{array}{ll} \ket{0},&\qquad \mbox{ with prob. }1-p^*, \,{\text{and}}\\
\ket{\sqrt{\E/p^*}}, & \qquad \mbox { with prob. } p^*,
\end{array}\right.
\ee 
where 
$\lim_{\E\to 0} \frac{p^*} {\frac{\E}{2} \log \frac{1}{\E}} = 1$,
and the resulting capacity is 
\be
\label{eqn:DD_approximation}
C_{\sf DD}(\E) = \E\log \frac{1}{\E} - \E \log \log \frac{1}{\E} + O(\E).
\ee}
\end{lemma}
\begin{proof}
Appendix \ref{sec:pf_lem_cap_DD}.
\end{proof}

Comparing~\eqref{eqn:holevo_approximation} and~\eqref{eqn:DD_approximation}, we observe that the two capacities have the same first-order term. This means as $\E\to 0$, the optimal photon information efficiency of $\log (1/\E)$ nats/photon can be achieved even with a very simple direct-detection receiver that acts directly and individually on each of the $N$ symbols of the $N$-mode modulated codeword.

In practice, however, the second-order terms in these two capacity expressions result in a significant difference in the high-photon-efficiency regime. For example, if one wishes to achieve a photon information efficiency of $10$ bits/photon, one can solve for $\E$ that satisfies $C(\E)/\E= 10$ bits/photon in both cases, and get $\E_{\sf Holevo} \approx 0.0027$ and $\E_{\sf DD} \approx 0.00010$. The resulting capacities (bits/mode, or equivalently the bits/sec-Hz spectral efficiencies) differ by more than one order of magnitude (by a factor of $\approx 26$ to be precise). So, if one is operating in a photon-starved regime, for instance in a deep space communications scenario where the mean photon number $\E$ per (temporal) mode is extremely small due to technological constraints on the transmit laser power and the large channel loss ($\eta \ll 1$), a Holevo-capacity-achieving receiver would attain more than an order of magnitude higher data rate for a given temporal bandwidth that can be supported by the transmit modulator and the receiver. This example indicates that although ~\eqref{eqn:holevo_approximation} and~\eqref{eqn:DD_approximation} have the same limit as $\E\to 0$, the rates at which this limit is approached are quite different, which can be of practical importance in photon-starved communication settings. Similar phenomena have also been observed for wideband wireless channels \cite{Verdu02, ZTM03}. 
 
Therefore, the second-order terms in the capacity expressions ~\eqref{eqn:holevo_approximation} and~\eqref{eqn:DD_approximation} cannot be ignored. In fact, any reasonable scheme that employs feedback-assisted coherent processing along with photon detection in the receiver should at the very least achieve a rate higher than that with direct detection alone, and thus should have the leading term as $\E\log \frac{1}{\E}$. It is the second-order term in the achievable rate that indicates whether a new receiver-structure proposal would make a significant step towards achieving the Holevo-capacity limit. In the following, we will study the achievable rate over the pure-loss optical channel with the Kennedy-Dolinar receiver front end as shown in Figure \ref{fig:dolinar}, and evaluate its rate performance and how it scales for small $\E$.

The problem of coded transmission and finding the maximum information rate that can be conveyed through an optical channel with a coherent-processing receiver is in fact easier than the problem of $M$-ary hypothesis testing we considered in Section~\ref{sec:gen_M}, even though there are exponentially many possible messages to discriminate between. The key observation is that when communicating with a long block of $N$ symbols (with $N \to \infty$), there is no issue of a pressing deadline for making a which-message decision for {\em most} of the time during the reception of a codeword. Therefore, it makes sense to always use the mutual information maximization to decide which control signal to apply. A straightforward generalization of the Dolinar receiver can be described as follows:

During the $i$-th channel use, $i \in \{1, \ldots, N\}$, the encoding map can be written $f_i: \{1, 2, \ldots, M=e^{NR}\} \to X_i \in {\cal X}$, where $X_i$ is the symbol transmitted in the $i$-th use of the channel. This map ensures that $X_i$ has a desired input distribution $P_X$, computed under the assumption that all messages are equally likely, i.e., $ \frac{1}{e^{NR}} | \{ m: f_i(m) = x\}| = P_X(x), \quad \forall x \in {\cal X}$. 

The receiver keeps track of the posterior distribution over the messages. Given $P_{M_s|Y_1^{i-1}}(\cdot | y_1^{i-1})$, which is the distribution over the messages conditioned on the previous observations, the effective input distribution when the receiver is about to act on the $i$-th channel symbol, $P'_X (x) = \sum_{m: f_i(m)=x} P_{M_s|Y_1^{i-1}}(m|y_1^{i-1})$ can be computed. Using this as the prior distribution of the transmitted symbol, the receiver can apply the control signal that maximizes the mutual information. 

Upon observing the output Poisson process in the $i$-th symbol period, denoted as $Y_i=y_i$, the receiver computes the posterior distribution of the transmitted symbol $P''_X(x) = P_{X_i|Y_i}(x|y_i)$. We omit the conditioning on the history $Y_1^{i-1}$ here to emphasize that the update is based on the observations in a single symbol period.
The receiver uses $P''_X(x)$ to update its knowledge of the messages in the following manner:
\be
P_{M_s|Y_1^i} (m|y_1^i)=P_{M_s|Y_1^{i-1}}(m|y_1^{i-1}) \cdot \frac{P''_X(x)}{P'_X(x)}
\ee
for all $m$ such that $f_i(m)=x$. This can be shown from
\be
\begin{split}
&P_{M_s|Y_1^i} (m|y_1^i)\\
&=P_{M_s|Y_1^{i-1}}(m|y_1^{i-1})\frac{P_{Y_i|M_s,Y_1^{i-1}}(y_i|m,y_1^{i-1})}{P_{Y_i|Y_1^{i-1}}(y_i|y_1^{i-1})} \\
&=P_{M_s|Y_1^{i-1}}(m|y_1^{i-1})\frac{P_{Y_i|X_i,Y_1^{i-1}}(y_i|x,y_1^{i-1})}{P_{Y_i|Y_1^{i-1}}(y_i|y_1^{i-1})} \\
&=P_{M_s|Y_1^{i-1}}(m|y_1^{i-1})\frac{P_{X_i|Y_i,Y_1^{i-1}}(x|y_i,y_1^{i-1})}{P_{X_i|Y_1^{i-1}}(x|y_1^{i-1})} 
\end{split}
\ee
Repeating this process, we have a coherent-processing receiver based on updating the receiver knowledge. 

There are two assumptions we make to simplify the analysis of capacity with a general coherent processing. Below are these assumptions. 

First, we assume that the control signal $l_i$ is kept constant within each symbol period (let us say, $\Delta$). Suppose that the $i$-th input symbol $X_i$ is transmitted over the symbol period $\Delta$. During this symbol period, the receiver would be able to continuously update the posterior distribution of $X_i$, which makes the effective input distribution deviate from the prior distribution. With the updated input distribution, the optimal control signal that maximizes the mutual information at each time instant might also change. But, here we assume that the control signal $l_i$ is determined at the beginning of each symbol period and kept constant during $\Delta$. 

Second, we will approximate the output Poisson process in each symbol period as a Bernoulli process, indicating either 0 or 1 photon arrival.
This assumption may not degrade the rate performance in a significant way when the mean photon number $\E$ per symbol is small enough.

The main result of our paper is the following theorem:

\begin{theorem}\label{thm:cap_coh}
\label{thm:sad}
{\it Consider a receiver front end as shown in Figure \ref{fig:dolinar}, and a control signal that is kept constant within each symbol of a codeword but updated from one symbol to the next. The photon counter at the receiver detects whether or not there are any photon arrivals within each symbol period. Suppose that the transmitted symbols are drawn from a finite alphabet, i.e., for the $i$-th channel, $i = 1, \ldots, N$, the transmitted optical signal $\ket{X_i}$ is chosen from $X_i \in {\cal X} \subset \mathbb{C}$ with $|{\cal X}|$ finite. Input symbols satisfy a mean-photon-number constraint $\mathbb{E}[|X_i|^2]=\E$ per  mode (per channel use). Then the achievable photon information efficiency (nats/photon) is bounded above as
\be
\label{eqn:sad}
\frac{C_{\sf coherent}(\E)}{\E} \leq \log \frac{1}{\E} - \log \log \frac{1}{\E} + O(1)
\ee  
when $\E\to0$.}
\end{theorem}
\begin{proof}
Appendix \ref{sec:proof_cap_coh}.
\end{proof}

Thus the achievable photon information efficiency with the Kennedy-Dolinar receiver front end is not significantly different from that of ideal direct detection alone. Note that despite the capacities in Eqs.~\eqref{eqn:sad} and ~\eqref{eqn:DD_approximation} being identical, the codes that the respective receiver may employ to attain this capacity may be very different.
\begin{figure}
\centering
\includegraphics[width=\columnwidth]{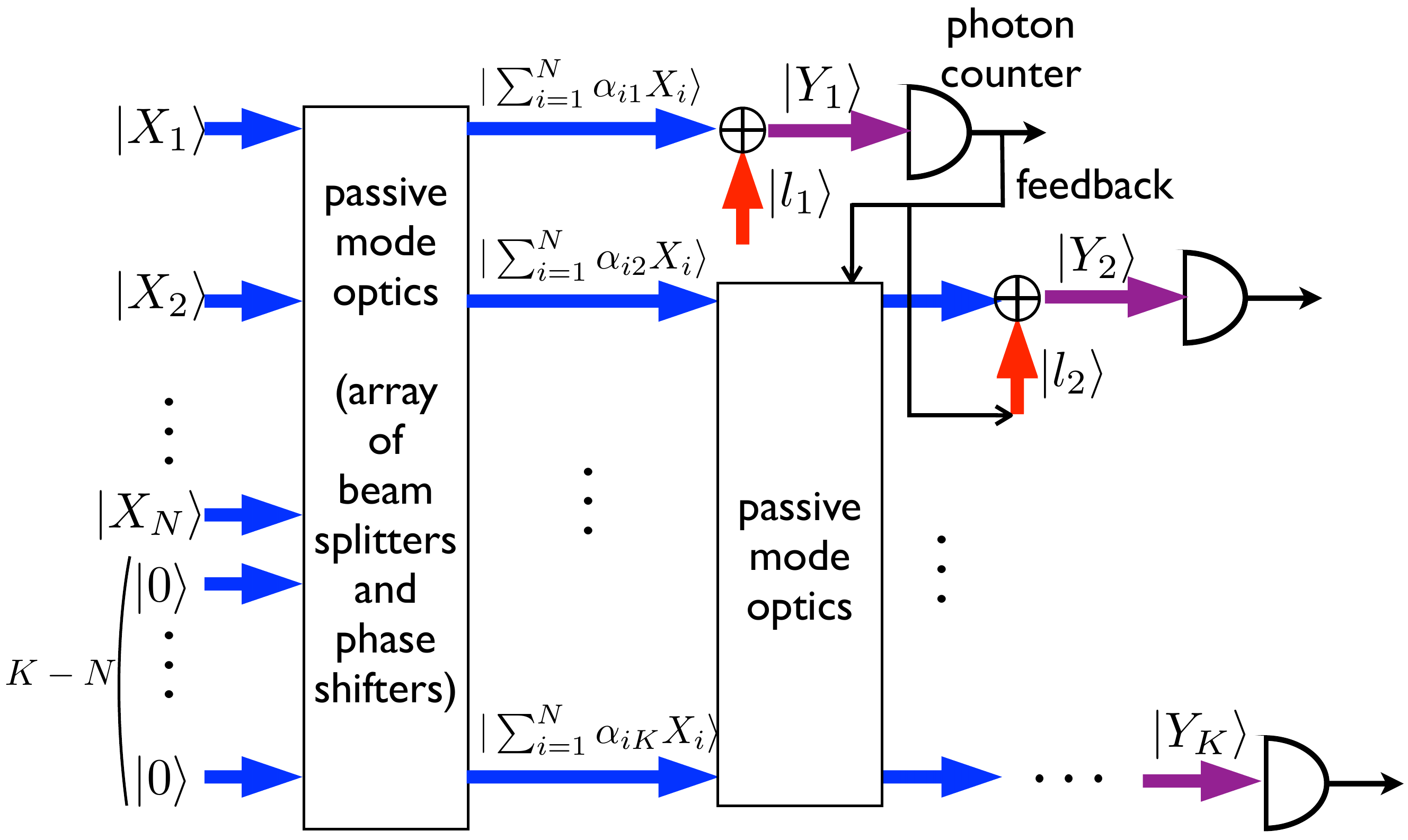}
\caption{General coherent-processing receiver with joint processing over multiple symbols. At the first stage, the receiver codeword  $\ket{X_1}\ket{X_2}\ldots\ket{X_N}$ augmented with $(K-N)$ auxiliary modes $\ket{0}\dots\ket{0}$ is processed by a set of beam splitters and phase shifters to generate a sequence of $K$ coherent states $\sum_{i=1}^N \alpha_{ij} X_i$, where $\sum_j |\alpha_{ij}|^2 \leq 1, \forall i$ and $\sum_i |\alpha_{ij}|^2 \leq 1, \forall j$. The receiver applies control signal $l_1$ to the first mixed signal $\sum_{i=1}^N \alpha_{i1} X_i$, to obtain $Y_1=\sum_{i=1}^N \alpha_{i1} X_i + l_1$ and detect the state with a photon counter. Given observations, a new set of parameters for the next passive mode transformation and the second control signal $l_2$ are determined. 
We repeat the similar process until all the $K$ output states are detected.
With this general coherent-processing receiver, the number $K$ of total output states detected at the receiver can be much larger than the number $N$ of received states.  }
\label{fig:general}
\end{figure}

This theorem is a useful step in understanding  the performance of a more general coherent-processing receiver with joint processing over multiple symbols. Let us consider the general receiver construct shown in Fig.~\ref{fig:general}, which is the natural generalization of the original Dolinar receiver idea as we describe below. The received codeword $\ket{X_1}\ket{X_2}\ldots\ket{X_N}$, where each $X_i$ is drawn from an alphabet, is processed by a general passive linear optical transformation---a circuit that can be composed of beamsplitters and phase shifters---to produce an $K$-mode product coherent state vector $\ket{Z_1}\ket{Z_2}\ldots\ket{Z_K}$, where ${\boldsymbol Z} = U_1{\boldsymbol X}$ with ${\boldsymbol Z} = [Z_1, Z_2, \ldots, Z_K]^{\rm T}$, ${\boldsymbol X} = [X_1, X_2, \ldots, X_N, 0, \ldots, 0]^{\rm T}$, and $U_1$ is a $K$-by-$K$ complex-valued unitary matrix. Fig.~\ref{fig:general} shows $K-N$ auxiliary modes at the input in a product of vacuum states. In the limit of infinite $K$, the mode transformation $U_1$ can produce arbitrarily-many output amplitudes that are each arbitrarily small. Thus the output sequence of $K$ coherent states have complex amplitudes, $\sum_{i=1}^N \alpha_{ij} X_i$, where $\sum_j |\alpha_{ij}|^2 \leq 1, \forall i$ and $\sum_i |\alpha_{ij}|^2 \leq 1, \forall j$, with equalities when the linear mode transformation is lossless. This translates to the physical constraint of energy conservation and the fact that duplication or noiseless amplification of coherent states is not possible. This action, a passive mode transformation, can always be broken down into $O(K^2)$ $2$-input $2$-output beamsplitters and phase shifters~\cite{Reck1994}. The receiver then applies an arbitrary control signal (coherent displacement) $l_1$ to the first output mode of $U_1$, to obtain $Y_1 =\sum_{i=1}^N \alpha_{i1} X_i + l_1$ and uses a photon detector to detect it. The detection outcome (a click or not) is then used to determine another linear mode transformation $U_2$ that mixes the $K-1$ remaining coherent states as well as to determine the coherent displacement $l_2$ applied to the first output mode produced by $U_2$ to produce $Y_2 =\sum_{i=1}^N \alpha'_{i2} X_i + l_2$, and so on. The receiver progressively detects output coherent states $\ket{Y_1}, \ket{Y_2}, \ldots, \ket{Y_K}$, while allowing for the control signals $l_j$ as well as the mixing parameters to be updated adaptively in each step based on the earlier observations. Note here that the original Dolinar receiver is a special case of this general receiver strategy (shown in Fig.~\ref{fig:general}) where the input is a one-mode ($N=1$) coherent state and each of the linear-optical mode transformations $U_1, U_2, \ldots$ are uniform mixers. One example of a uniform mixer is the linear-optical Hadamard unitary, considered in~\cite{Guh10}. 

Following the spirit of Theorem \ref{thm:sad}, we state the following conjecture.

\begin{conjecture}\label{conj:sad}
{\it The maximum achievable photon information efficiency using an optical receiver as shown in Figure~\ref{fig:general}---a collective-measurement multi-mode generalization of the Dolinar receiver---is given by ~\eqref{eqn:sad}.}
\end{conjecture}

While this conjecture is a negative one, it is of immense practical importance in understanding the power of linear optical processing and photon detection, and may have implications to other applications of quantum-limited optical processing such as in linear optical quantum computing (LOQC). Even though the codewords being discriminated are a product (sequence) of (classical) coherent states, the optimal capacity-achieving receiver must use non-classical joint processing over the modulated codeword prior to detecting it. We believe that~\eqref{eqn:sad} quantifies the ultimate rate performance achievable by absolutely any optical receiver whose workings can be described quantitatively correctly using the semi-classical (shot noise) theory of photo detection. Direct detection without any feedback or coherent pre-processing can already attain this performance. This conjecture's truth would imply that in order to achieve the photon information efficiency predicted by the Holevo limit, it would be necessary to use truly quantum processing within the receiver. Examples of such actions include replacing the coherent-state local control signals with squeezed states, or mixing the received codeword with a locally prepared $N$-mode entangled state prior to detection. In order to analyze such receivers, we can no longer use shot-noise (Poisson-limited) noise models, and must resort to the full quantum theory of photo detection. 

In recent work, Rosati {\em et al.} proved the aforesaid conjecture for a receiver structure we consider above, but restricted to the case of no auxiliary vacuum modes, i.e., $U_1$ acting on $N$ modes, $U_2$ on $N-1$ modes, and so on~\cite{rosati2017}. It will be interesting to consider whether their proof technique applies to the more general case.

Finally, we would like to note that even though we believe that the receiver structure described in Conjecture~\ref{conj:sad} (a collective-measurement multi-mode generalization of the Dolinar receiver) is ineffective in attaining capacity that is any better than what ideal direct detection alone can, this type of all-optical pre-processing can immensely lessen the peak-power requirements compared to the high-peak-power OOK modulation that must be used by the direct-detection receiver to attain rate performance as stated in~\eqref{eqn:DD_approximation}. An example of such a receiver was described in~\cite{Guh10}, using which a binary-phase-shift-keying modulation (which has the minimum possible peak power in the $\E \ll 1$ regime) could achieve the same rate scaling as in~\eqref{eqn:DD_approximation}. The scheme in~\cite{Guh10} uses a passive linear-mode mixing on the codeword symbols, but does not use any local signals prior to detection. In order to attain $10$ bits/photon using OOK (or, pulse-position) modulation with direct detection, one would require roughly $3$ orders of magnitude higher peak power compared to this scheme. For deep-space communications, reduction in the peak laser-power requirement could translate to much longer ranges being made possible.

\section{Conclusion}\label{sec:conclusion}
We studied the general coherent-state hypothesis-testing problem and the capacity of the pure-loss optical channel with a general {\em coherent-processing} receiver---a receiver that uses ideal direct detection, and coherent electro-optic feedback control that mixes a coherent-state local oscillator with the incoming signal while it is being detected. We  re-interpreted Dolinar's receiver for optimally discriminating binary coherent-state hypotheses as an instantaneous optimization of the communication efficiency using recursively-updated knowledge based on the observed photon-arrival events. Using this viewpoint, we presented a natural generalization of Dolinar's receiver design to the general $M$-ary coherent-state hypothesis-testing problem. We analyzed the information capacity attained with this generalized Kennedy-Dolinar receiver front end (shown in Figure \ref{fig:dolinar}), and compared the result with that of an ideal direct-detection receiver (with no internal feedback or coherent processing) as well as to that achievable by an unconstrained quantum-limited joint-detection receiver (the Holevo limit), using appropriate scalings in the low photon-number-per-mode regime. 

Our main result in Theorem~\ref{thm:cap_coh} is a negative result, but is of practical importance. It implies that in order to achieve the photon information efficiency predicted by the Holevo limit, it is necessary to resort to truly quantum-limited processing that may include using entanglement or squeezing locally within the receiver, despite the fact that the state of the codeword being demodulated is completely classical. We conjectured that no semi-classical receiver strategy, even one that mixes the received codeword symbols using an arbitrary circuit of passive elements prior to applying adaptive local control signals, would  yield any significant performance improvement over direct detection. Finally, we argued that even if the aforesaid conjecture is true, coherent pre-processing and electro-optic coherent-feedback-control-based optical receiver can immensely reduce the strain on the transmitter and coding fronts, for instance by reducing the peak-transmit-power requirements over a highly lossy optical channel.

\begin{acknowledgments}
{This research was supported by the Defense Advanced Research Projects Agency's (DARPA) Information in a Photon (InPho) program under a contract (\#HR0011-10-C-0159) to Raytheon BBN Technologies, with a subcontract to MIT. SG would like to thank Sam Dolinar and Mark Neifeld for many useful discussions on this topic. The views and conclusions contained in this document are those of the authors and should not be interpreted as representing the official policies, either expressly or implied, of DARPA or the U.S. Government.}
\end{acknowledgments}


\appendix

\section{Proof of Lemma \ref{lem:opt_feed}}\label{sec:pf_lem_opt_feed}
In Lemma \ref{lem:opt_feed}, we show that the optimal choice of the control signal $l$ of Dolinar receiver that maximizes the mutual information $I(H;Y)$ between binary hypothesis $H\in\{0,1\}$ and receiver output $Y\in\{0,1\}$ equals
\be\label{eqn:l^*_appe}
l^*=\frac{S_0\pi_0-S_1\pi_1}{\pi_1-\pi_0},
\ee 
where $\{\pi_0,\pi_1\}$ and $\{S_0,S_1\}$ are input probabilities and signal amplitudes for hypothesis $H\in\{0,1\}$, respectively. 
The channel distribution $P_{Y|H}$ between hypothesis $H$ and receiver output $Y$ is
\be
\begin{split}
P_{Y|H}(j|i)=\begin{cases}
 e^{-\lambda_i\Delta},&j=0,\\
 1-e^{-\lambda_i\Delta},&j=1,\\
\end{cases}
\end{split}
\ee
 where $\lambda_i=|S_i+l|^2$ for $i=0,1$.
The mutual information $I(H;Y)$ of this channel with input probabilities $\{\pi_0,\pi_1\}$ equals
\begin{equation}
\begin{split}
&I(H;Y)\\
=&\pi_0\left(e^{-\lambda_0\Delta}\log\frac{e^{-\lambda_0\Delta}}{\pi_0 e^{-\lambda_0\Delta}+\pi_1e^{-\lambda_1\Delta}}\right.\\
&\quad\quad\left.+\left(1-e^{-\lambda_0\Delta}\right)\log\frac{1-e^{-\lambda_0\Delta}}{1-\pi_0 e^{-\lambda_0\Delta}-\pi_1 e^{-\lambda_1\Delta}}\right)\\
&+\pi_1\left(e^{-\lambda_1\Delta}\log\frac{e^{-\lambda_1\Delta}}{\pi_0 e^{-\lambda_0\Delta}+\pi_1e^{-\lambda_1\Delta}}\right.\\
&\qquad\quad\left.+\left(1-e^{-\lambda_1\Delta}\right)\log\frac{1-e^{-\lambda_1\Delta}}{1-\pi_0 e^{-\lambda_0\Delta}-\pi_1 e^{-\lambda_1\Delta}}\right).
\end{split}
\end{equation}
As $\Delta\to0$, this mutual information can be approximated as
\begin{equation}
\begin{split}\label{eqn:expofIHY}
&I(H;Y)=\\
&\left(\pi_0\lambda_0\log\lambda_0+\pi_1\lambda_1\log\lambda_1\right.\\
&\quad\left. -(\pi_0\lambda_0+\pi_1\lambda_1)\log(\pi_0\lambda_0+\pi_1\lambda_1)\right)\Delta+O(\Delta^2).
\end{split}
\end{equation}

With the control signal $l^*$ in~\eqref{eqn:l^*_appe}, the mutual information $I(H;Y)$ is equal to
\be\label{eqn:MI_woptl}
I(H;Y)|_{l=l^*}=\left(\frac{(S_0-S_1)^2\pi_0\pi_1}{\pi_1-\pi_0}\log\frac{\pi_1}{\pi_0}\right)\Delta+O(\Delta^2).
\ee 

We next show that with any other value for the control signal $l$, the resulting mutual information cannot exceed the right-hand side of~\eqref{eqn:MI_woptl}.
To show this, we use the results in~\cite{holevo1998quantum,sohma2000binary} that when binary input states of amplitudes $\{{S_0},{S_1}\}$ with probabilities $\{\pi_0,\pi_1\}$ are measured by any single-symbol (unentangling) measurement, the resulting mutual information is bounded above by
\be\label{up_IHY}
I(H;Y)\le H_{\sf B}\left(\pi_0\right)-H_{\sf B}\left(P_e^*\right)
\ee
where
\be
\begin{split}\label{eqn:P_e^*_appen}
&H_{\sf B}(p)=-p\log p-(1-p)\log(1-p),\\
&P_e^*=\frac{1-\sqrt{1-4\pi_0\pi_1e^{-(S_0-S_1)^2\Delta}}}{2}.
\end{split}
\ee
As $\Delta\to 0$, $P_e^*$ in~\eqref{eqn:P_e^*_appen} can be approximated as
\be
\begin{split}
&P_e^*=\\
&\frac{1}{2}\left(1-|\pi_0-\pi_1|\left(1+\frac{2\pi_0\pi_1(S_0-S_1)^2\Delta}{(\pi_0-\pi_1)^2}\right)\right)+O(\Delta^2).
\end{split}
\ee
By using this, we can show that the right hand side of~\eqref{up_IHY} is
\be
\begin{split}
&H_{\sf B}(\pi_0)-H_{\sf B}(P_e^*)\\
&=\left(\frac{(S_0-S_1)^2\pi_0\pi_1}{\pi_1-\pi_0}\log\frac{\pi_1}{\pi_0}\right)\Delta+O(\Delta^2).
\end{split}
\ee
This proves that $l^*$ in~\eqref{eqn:l^*_appe} is the optimal choice of $l$ that maximizes $I(H;Y)$ in~\eqref{eqn:expofIHY}.

\section{Proof of Lemma \ref{lemma:binaryoptimality}}\label{sec:pf_lem_opt_feed1}
In Lemma~\ref{lemma:binaryoptimality}, we find the optimal control signal $l^*(t)$ of the coherent receiver from the recursive mutual-information-maximization procedure and show that the resulting probability of error for binary hypothesis testing achieves the YKL limit~\cite{YKL1975, Hel76}, the lower bound on the detection error probability over all possible quantum receivers.

To find the optimal control signal $l(t)$ over time $t$, we consider $S_0(t)$, $S_1(t)$ and $l(t)$ for each infinitesimal interval $t\in[k\Delta,(k+1)\Delta)$ of length $\Delta>0$ for $k\in\{0,1,\dots\}$.
When $S_0$ and $S_1$ denote the constant values of $S_0(t)$ and $S_1(t)$, respectively, for a very small interval $t\in[k\Delta,(k+1)\Delta)$, the optimal control signal $l^*$ that maximizes the mutual information between input hypothesis of probabilities $\{\pi_0$,$\pi_1\}$ and receiver output over the symbol period $\Delta$ is
\be
l^*=\frac{S_0\pi_0-S_1\pi_1}{\pi_1-\pi_0}
\ee
as shown in Lemma~\ref{lem:opt_feed}.
When we choose the control signal $l(t)$ by recursive mutual-information-maximazation procedure and make $\Delta\to0$, the optimal control signal becomes
\be\label{eqn:appBl^*}
l^*(t)=\frac{S_0(t)\pi_0(t)-S_1(t)\pi_1(t)}{\pi_1(t)-\pi_0(t)}
\ee
where $\pi_0(t)$ and $\pi_1(t)$ are posterior probabilities over the two hypotheses, conditioned on the trace of output of the coherent receiver until time $t$.
The question is then how the two posterior probabilities $\pi_0(t)$ and $\pi_1(t)$ evolve over time $t$.

We first focus on the first length-$\Delta$ interval, i.e, $t\in[0,\Delta)$, and find $\pi_0(\Delta)$ and $\pi_1(\Delta)$. Define $\pi_0:=\pi_0(0)$, $\pi_1:=\pi_1(0)$ and assume that $\pi_0\geq \pi_1$ without loss of generality. We define 
\begin{equation}\label{eqn:appBgt}
g(t):=\max\{\pi_0(t)/\pi_1(t), \pi_1(t)/\pi_0(t)\}.
\end{equation}
Note that $g(0)=\pi_0/\pi_1\geq 1$.
When the output of the receiver during the first $\Delta$ interval is denoted as $Y_0\in\{0,1\}$, for $Y_0=0$
\begin{equation}
\begin{split}\label{eqn:post_1st}
&\frac{\Pr(H=0|Y_0=0)}{\Pr(H=1|Y_0=0)}=\frac{\pi_0}{\pi_1}\cdot\frac{\Pr(Y_0=0|H=0)}{\Pr(Y_0=0|H=1)}\\
&=\frac{\pi_0}{\pi_1}\cdot\frac{e^{-(S_0(0)+l(0))^2\Delta}}{e^{-(S_1(0)+l(0))^2\Delta}}.
\end{split}
\end{equation}
By plugging the optimal control signal $l(0)$,
\begin{equation}\label{eqn:optl0}
\begin{split}
l(0)&=\frac{S_0(0)\pi_0-S_1(0)\pi_1}{\pi_1-\pi_0}\\
&=\frac{S_1(0)-S_0(0) g(0) }{g(0)-1},
\end{split}
\end{equation}
which maximizes the mutual information over the first symbol period $\Delta$, we obtain 
\begin{equation}\label{eqn:post_y0}
\begin{split}
&\frac{\Pr(H=0|Y_0=0)}{\Pr(H=1|Y_0=0)}=\frac{\pi_0}{\pi_1}\cdot e^{\left(S_0(0)-S_1(0)\right)^2\frac{g(0)+1}{g(0)-1}\Delta}.
\end{split}
\end{equation}
Note that $\Pr(H=0|Y_0=0)/\Pr(H=1|Y_0=0)\geq \pi_0/\pi_1$ since $g(0)\geq 1$.

When $Y_0=1$, on the other hand, the ratio between the two posterior probabilities becomes
\begin{equation}
\begin{split}\label{eqn:post_2nd}
&\frac{\Pr(H=0|Y_0=1)}{\Pr(H=1|Y_0=1)}=\frac{\pi_0}{\pi_1}\cdot\frac{\Pr(Y_0=1|H=0)}{\Pr(Y_0=1|H=1)}\\
&=\frac{\pi_0}{\pi_1}\cdot\frac{1-e^{-(S_0(0)+l(0))^2\Delta}}{1-e^{-(S_1(0)+l(0))^2\Delta}}.
\end{split}
\end{equation}
As $\Delta\to0$,
\begin{equation}
\begin{split}\label{eqn:post_y1}
&\frac{\Pr(H=0|Y_0=1)}{\Pr(H=1|Y_0=1)}=\frac{\pi_0}{\pi_1}\cdot\frac{(S_0(0)+l(0))^2}{(S_1(0)+l(0))^2}+O(\Delta)\\
&=\frac{\pi_1}{\pi_0}+O(\Delta).
\end{split}
\end{equation}
The ratio between the two posterior probabilities conditioned on $Y_0=1$ in (\ref{eqn:post_y1}) is approximately inverse of that conditioned on  $Y_0=0$ in (\ref{eqn:post_y0}). 
Therefore, $g(t)$ in~\eqref{eqn:appBgt}, indicating how much the receiver is committed to the more likely hypothesis, is uniquely determined and increases at a prescribed 
 rate regardless of photon arrivals over time $[0,t)$.

To find how $g(t)$ evolves over time $t$, without loss of generality we focus on a particular case where no photon arrives during $[0,t)$. 
From (\ref{eqn:post_y0}), 
\begin{equation}
g(\Delta)=\frac{\pi_0}{\pi_1}\cdot e^{\left(S_0(0)-S_1(0)\right)^2\frac{g(0)+1}{g(0)-1}\Delta}.
\end{equation}
Under the assumption that no photon arrives for the next $(N-1)$ intervals, i.e., for the sequence of all-zero outputs $Y_1=\dots=Y_{N-1}=0$,  we obtain the following recursive equation for $g(N\Delta)$:
\begin{equation}
\begin{split}
&g(N \Delta)=\frac{\Pr(H=0|Y_0^{N-1}={\mathbf0})}{\Pr(H=1|Y_0^{N-1}={\mathbf0})}\\
&=\frac{\pi_0}{\pi_1}e^{\left(\sum_{k=0}^{N-1} \left(\left(S_0(k \Delta)-S_1(k \Delta)\right)^2\frac{g(k \Delta)+1}{g(k\Delta)-1}\Delta\right)\right)}.
\end{split}
\end{equation}
By taking $\Delta\to 0$, we obtain
\begin{equation}
\begin{split}\label{eqn:gt_rec}
g(t)=&\frac{\pi_0}{\pi_1} \exp\left[\int_{0}^{t} \left(\left(S_0(\tau)-S_1(\tau)\right)^2\cdot \frac{g(\tau)+1}{g(\tau)-1} \right)d\tau \right]\\
=&g(0)\exp\left[\int_{0}^{t} \left(\left(S_0(\tau)-S_1(\tau)\right)^2\cdot \frac{g(\tau)+1}{g(\tau)-1} \right)d\tau \right].
\end{split}
\end{equation}

Let $N(t)$ be the number of photon arrivals observed during $[0,t)$. We showed that whenever a photon arrives at the receiver, the ratio $\pi_0(t)/\pi_1(t)$ between two posterior probabilities   gets flipped. Therefore, starting from $g(0)=\pi_0/\pi_1\geq 1$, $g(t)$ defined in~\eqref{eqn:appBgt} equals $\pi_0(t)/\pi_1(t)$  if $N(t)$ is even, and equals $\pi_1(t)/\pi_0(t)$  if $N(t)$ is odd.
By using this relation, the optimal control signal $l^*(t)$ in~\eqref{eqn:appBl^*} can be written in terms of $g(t)$ as
\be
l^*(t) = \left\{ 
\begin{array}{ll}
l_0(t) & \quad \mbox{if } N(t) \mbox{ is even}\\
l_1(t) & \quad \mbox{if } N(t) \mbox{ is odd}
\end{array}
\right.
\ee
where 
\be
l_0(t) = \frac{S_1(t)- S_0(t) g(t)}{g(t)-1}, \quad
l_1(t) = \frac{S_0(t)- S_1(t) g(t)}{g(t)-1}.
\ee
Furthermore, the final decision of more likely hypothesis at $t=T$ is $\widehat{H}=0$ if $N(T)$ is even, and $\widehat{H}=1$ otherwise.
The average probability of error is then equal to $P_e=\min\{\pi_0(t),\pi_1(t)\}$, and by the definition of $g(t)$,
\begin{equation}
P_e=\frac{1}{1+g(t)}.
\end{equation}
When we solve the recursive equation on $g(t)$ in~\eqref{eqn:gt_rec}, we obtain
\begin{equation}
\begin{split}
g(t)=&\frac{(1+g(0))^2}{2g(0)}e^{m(t)}-1\\
&+\frac{1+g(0)}{2g(0)}\sqrt{(1+g(0))^2 e^{2m(t)}-4g(0)e^{m(t)}}
\end{split}
\end{equation}
where $m(t)=\int_0^{t} (S_0(\tau)-S_1(\tau))^2d\tau$.
The resulting $P_e$ is 
\begin{equation}
\begin{split}
P_e=&\frac{1}{1+g(t)}\\
=&\frac{1}{2}\left(1-\sqrt{1-4\pi_0\pi_1e^{-\int_0^{t} (S_0(\tau)-S_1(\tau))^2d\tau}}\right),
\end{split}
\end{equation}
which is equal to the YKL limit. 

\section{Proof of Lemma \ref{lem:cap_DD}}\label{sec:pf_lem_cap_DD}
In Lemma \ref{lem:cap_DD}, we show that the capacity of optical channel with direction direction is
\be\label{C1}
C_{\sf DD}(\E)=\E\log\frac{1}{\E}-\E\log\log\frac{1}{\E}+O(\E)
\ee
where $\E$ is the mean photon number per channel use.
This capacity is achievable with on-off keying inputs
\be\label{eqn:on-offkeying_inputs}
\ket{S} = \left\{ \begin{array}{ll} \ket{0},&\qquad \mbox{ with prob. }1-p^*\\
\ket{\sqrt{\E/p^*}}, & \qquad \mbox { with prob. } p^*
\end{array}\right.
\ee 
where 
$\lim_{\E\to 0} \frac{p^*} {\frac{\E}{2} \log \frac{1}{\E}} = 1$.

The converse part of this lemma, i.e., that the capacity of optical channel with direction detection can never exceed
\begin{equation}\label{eqn:DD_cap_OOK}
\E\log\frac{1}{\E}-\E\log\log\frac{1}{\E}+O(\E),
\end{equation}
is implied from the converse proof of Theorem \ref{thm:cap_coh}, which considers a more general receiver type, which makes the direct detection as a special case.

Here we prove the achievability of the capacity in (\ref{C1}) with on-off-keying inputs~\eqref{eqn:on-offkeying_inputs}.
When direct-detection receiver measures the off signal, i.e., $\ket{S}=\ket{0}$, which is transmitted with probability $1-p^*$, the output of direction-detection receiver, which counts the number of photon arrivals per symbol period, equals 0 with probability 1. On the other hand, when on-signal $\ket{S}=\ket{\sqrt{\E/p^*}}$ is transmitted with probability $p^*$, the direction-detection receiver observes 0 photon with probability $e^{-\E/p^*}$ and at least 1 photon with probability $1-e^{-\E/p^*}$. The mutual information between the on-off keying input $S$ and binary output $Y$ of the direction-detection receiver equals
\begin{equation}
\begin{split}\label{C4}
&I(S;Y)=H_{\sf B}\left(p^*\left(1-e^{-\frac{\E}{p^*}}\right)\right)-p^* H_{\sf B}\left(1-e^{-\frac{\E}{p^*}}\right)
\end{split}
\end{equation}
where $H_{\sf B}(p)=-p\log p-(1-p) \log(1-p)$.

For ${p^*} ={\frac{\E}{2} \log \frac{1}{\E}} $, by using the Taylor expansion,  we can approximate
\begin{equation}
\begin{split}
&1-e^{-\frac{\E}{p^*}}=\frac{2}{\log(1/\E)}+O\left(\frac{1}{(\log(1/\E))^2}\right),\\
&p^*\left(1-e^{-\frac{\E}{p}}\right)=\E+O\left(\frac{1}{(\log(1/\E))}\right),\\
\end{split}
\end{equation}
as $\E\to0$.
By using these approximations and $H_{\sf B}(q)=-q\log q+q+O(q^2)$ as $q\to0$, we can show that
\begin{equation}
\begin{split}\label{C6}
& H_{\sf B}\left(p^*\left(1-e^{-\frac{\E}{p^*}}\right)\right)=\E\log\frac{1}{\E}+O(\E),\\
& p^* H_{\sf B}\left(1-e^{-\frac{\E}{p^*}}\right)=\E\log\log\frac{1}{\E}+O(\E).
\end{split}
\end{equation}
From~\eqref{C4} and~\eqref{C6}, we obtain
\begin{equation}
\begin{split}
I(S;Y)=\E\log\frac{1}{\E}-\E\log\log\frac{1}{\E}+O(\E).
\end{split}
\end{equation}
By combining with the converse part, this achievability result implies~\eqref{C1}.

\section{Proof of Theorem \ref{thm:cap_coh}}\label{sec:proof_cap_coh}



In Theorem~\ref{thm:cap_coh}, we show that the achievable photon information efficiency for pure-state optical channels with coherent-processing receiver is bounded above by
\be\label{eqn:sad?}
\frac{C_{\sf coherent}(\E)}{\E} \leq \log \frac{1}{\E} - \log \log \frac{1}{\E} + O(1)
\ee  
where $\E$ is the mean-photon-number constraint for the input coherent state $\ket{X_i}$, $X_i\in\mathcal{X}\subset\mathbb{C}$, for finite $|\mathcal{X}|$, i.e.,
\be
\mathbb{E}[|X_i|^2]\leq\E.
\ee
From Lemma~\ref{lem:cap_DD}, it can be easily shown that the equality in~\eqref{eqn:sad?} is achievable with coherent-processing receiver, since coherent-processing receiver is equivalent to direct-detection receiver when the control signal is fixed to 0 over all communication periods, and Lemma~\ref{lem:cap_DD} shows that the right hand side of~\eqref{eqn:sad?} is achievable with  the direct-detection receiver for on-off-keying input signaling. The remaining thing to show is the converse part of the theorem, i.e.,  the claim that with coherent-processing receiver one can never achieve photon information efficiency better than the right hand side of~\eqref{eqn:sad?}.

Suppose that a message is chosen from a set $\{1,\dots, e^{NR}\}$ with equal probabilities and is transmitted by $N$ uses of the optical channel. The $i$-th transmitted optical signal (coherent state) is denoted by $\ket{X_i}$, $X_i\in\mathcal{X}\subset\mathbb{C}$, and the associated output of the coherent-processing receiver is denoted by $Y_i\in\{0,1\}$, indicating 0 or 1 photon arrival during a very short symbol period. We use the notation $Y_{i}^j$, $j>i$, to indicate a sequence of output random variables $(Y_i,Y_{i+1},\dots, Y_j)$. When $M_s$ and $\hat{M_s}(Y_1^N)$ denote the transmitted message and the estimate of it based on the output sequence $Y_1^N$, respectively, decoding error probability after $N$ uses of the channel is defined as 
\be
P_e^{(N)}=\Pr(M_s\neq \hat{M_s}(Y_1^N)).
\ee
From Fano's inequality~\cite{cover2012elements}, the decoding error probability $P_e^{(N)}$ is bounded below as
\be
P_e^{(N)}\geq 1-\frac{I(X_1^N;Y_1^N)}{NR}-\frac{\ln 2}{NR}.
\ee
If $R>\frac{I(X_1^N;Y_1^N)}{N}$, this lower bound is larger than 0, meaning that $P_e^{(N)}$ does not converge to 0 even when $N\to\infty$.
Therefore, the capacity $C_{\sf coherent}(\E)$ of coherent-processing receiver, which is the maximum information rate that guarantees $P_e^{(N)}\to 0$ as $N\to \infty$, is bounded above by 
\be\label{eqn:ub_C_co}
C_{\sf coherent}(\E)\leq \frac{I(X_1^N;Y_1^N)}{N}.
\ee
We next find an upper bound on $I(X_1^N;Y_1^N)$. First note that 
\be
\begin{split}\label{eqn:chain+}
I(X_1^N;Y_1^N)&=\sum_{i=1}^N \left(H(Y_i|Y_1^{i-1})-H(Y_i|X_1^N,Y_1^{i-1})\right)\\
&=\sum_{i=1}^N \left(H(Y_i|Y_1^{i-1})-H(Y_i|X_i,Y_1^{i-1})\right)\\
&=\sum_{i=1}^N I(X_i;Y_i|Y_1^{i-1})\\
&=\sum_{i=1}^N\mathbb{E}_{Y_1^{i-1}}[I(X_i;Y_i|Y_1^{i-1}=y_1^{i-1})],
\end{split}
\ee
where the first equality is from the chain rule and definition of the mutual information, and the second equality is from the fact that $Y_i$ is independent of $\{X_1^{i-1}, X_{i+1}^N\}$ conditioned on the $i$-th input $X_i$ and the past observations $Y_1^{i-1}$. The third and the fourth equalities are from the definition of the conditional mutual information $I(X_i;Y_i|Y_1^{i-1})$. 

We next provide an upper bound on $I(X_i;Y_i|Y_1^{i-1}=y_1^{i-1})$, which is independent of $Y_1^{i-1}=y_1^{i-1}$.
Since the transmitter does not know the past channel outputs $Y_1^{i-1}=y_1^{i-1}$ at the receiver, the $i$-th input symbol $X_i$ is independent of  $Y_1^{i-1}=y_1^{i-1}$. On the other hand, the $i$-th output symbol $Y_i$ depends not only on the $i$-th input $X_i$ but also on the past channel outputs $Y_1^{i-1}=y_1^{i-1}$ through the control signal $l_i(y_1^{i-1})$ as
\be
\begin{split}\label{eqn:newBB}
&\Pr(Y_i=0|X_i, Y_1^{i-1}=y_1^{i-1})=e^{-|X_i+l_i(y_1^{i-1})|^2},\\
&\Pr(Y_i=1|X_i,Y_1^{i-1}=y_1^{i-1})=1-e^{-|X_i+l_i(y_1^{i-1})|^2},
\end{split}
\ee
for $X_i\in\mathcal{X}\subset\mathbb{C}$.
Here, for simplicity, we subsume the symbol period $\Delta$ into the input signal $X_i$ and the control signal $l_i$, i.e., for symbol period $\Delta$, complex field amplitudes of the input and the control signal are kept constant as $X_i/\sqrt{\Delta}$ and $l_i/\sqrt{\Delta}$, respectively. Due to the constraint on mean photon number per channel use, the input random variable $X_i$ in~\eqref{eqn:newBB} should satisfy $\mathbb{E}[|X_i|^2]\leq\E$.

For a complex constant value $l$, which is fixed during a symbol period $\Delta$,  define a channel distribution $P_{Y|X}$ such that
\be
\begin{split}\label{eqn:newB}
&P_{Y|X}(Y=0|X)=e^{-|X+l|^2},\\
&P_{Y|X}(Y=1|X)=1-e^{-|X+l|^2}.
\end{split}
\ee
When we define $I_l(P_X, P_{Y|X})$ as the mutual information between $X$ and $Y$ with input distribution $P_X$ and channel distribution $P_{Y|X}$ in~\eqref{eqn:newB}, the conditional mutual information $I(X_i;Y_i|Y_1^{i-1}=y_1^{i-1})$ with some input distribution $P_{X_i}$ and channel distribution~\eqref{eqn:newBB} is bounded above as
\be\label{eqn:singleI}
I(X_i;Y_i|Y_1^{i-1}=y_1^{i-1})\leq \max_{P_X, l} I_l(P_X, P_{Y|X}).
\ee
From~\eqref{eqn:chain+} and ~\eqref{eqn:singleI}, we obtain
\be
I(X_1^N;Y_1^N)\leq N \left(\max_{P_X, l} I_l(P_X, P_{Y|X})\right),
\ee
which implies 
\be
C_{\sf coherent}(\E)\leq \max_{P_X, l} I_l(P_X, P_{Y|X})
\ee
from~\eqref{eqn:ub_C_co}.

We next show that $\max_{P_X, l} I_l(P_X, P_{Y|X})$ is bounded above by
\be\label{eqn:core_ineq}
\max_{P_X, l} I_l(P_X, P_{Y|X})\leq \E\log \frac{1}{\E} -\E \log \log \frac{1}{\E} + O(\E).
\ee
To show this, we use the mathematical induction. We first show that for every binary input states, i.e., when $|\mathcal{X}|=2$, the bound~\eqref{eqn:core_ineq} holds.
We next assume that the bound~\eqref{eqn:core_ineq} holds when the input set $\mathcal{X}$ is constrained to have $L$ number of elements, i.e., when $|\mathcal{X}|=L$.
We then show that the same bound holds when $|\mathcal{X}|=L+1$. This will imply that the bound~\eqref{eqn:core_ineq} holds for any finite $|\mathcal{X}|$.

Let $R_L(\E)$ denote $\max_{P_X, l} I_l(P_X, P_{Y|X})$ under the constraint on the cardinality of the input set $|\mathcal{X}|=L$, i.e.,
\be\label{eqn:RLEdef}
R_L(\E):=\max_{|\mathcal{X|}=L}\left(\max_{P_X, l} I_l(P_X, P_{Y|X})\right).
\ee 
We first show that
\be\label{eqn:R2c}
R_2(\E)\leq \E\log \frac{1}{\E} -\E \log \log \frac{1}{\E} + O(\E)
\ee
for $|\mathcal{X}|=2$.
This bound can be implied by using Lemma 1 in~\cite{chung2016}.
Lemma 1 in~\cite{chung2016} shows that when binary-input coherent state with mean-photon-number constraint of $\E$ is detected by optimal single-symbol receiver measurement, which maximizes the mutual information of the induced channel, the resulting maximum mutual information is bounded above by the the right hand side of~\eqref{eqn:R2c}.  The coherent-processing receiver, which is composed of a mixture of a feedback signal followed by the direction-detection receiver, is a special case of the single-symbol receiver measurement. Therefore, Lemma 1 in~\cite{chung2016} implies that the bound in~\eqref{eqn:R2c} holds for every binary-input states of mean photon number $\E$ detected by the coherent-processing receiver. 

We next show that the same upper bound holds for $R_{L+1}$, i.e.,
\be\label{eqn:RL+1c}
R_{L+1}(\E) \leq \E\log \frac{1}{\E} -\E \log \log \frac{1}{\E} + O(\E),
\ee
when we assume that
\be\label{eqn:RLc}
R_{L}(\E) \leq \E\log \frac{1}{\E} -\E \log \log \frac{1}{\E} + O(\E).
\ee 



 
We first consider real-valued input signals, i.e., $X\in\mathcal{X}\subset\mathbb{R}$, and then later generalize the result for complex-valued input signals. 
For a fixed feedback control signal $l\in\mathbb{R}$ and the input set $\mathcal{X}=\{S_1',\dots, S_{L+1}'\}\subset\mathbb{R}$, without loss of generality, we can rearrange those $(L+1)$ amplitudes such that
\begin{equation}\label{eqn:assumption_M+1}
 |S_1+l|^2\leq \dots\leq |S_{L+1}+l|^2.
\end{equation}
 We denote the input distribution over the re-arranged input set $\{{S_1},\dots, {S_{L+1}}\}$ as $\{p_1,\cdots,p_{L+1}\}$, i.e., $\Pr(X=S_i)=p_i$.
 The resulting mutual information for the given input distribution and a fixed $l$ is 
 \be
 \begin{split}
& I_l(P_X, P_{Y|X})\\
& =H_{\sf B}\left(\sum_{i=1}^{L+1}p_i e^{-|S_i+l|^2}\right)-\sum_{i=1}^{L+1} p_i H_{\sf B}\left(e^{-|S_i+l|^2}\right)
 \end{split}
 \ee
 where the entropy $H_{\sf B}(p)$ for some Bernoulli random variable $Z\sim \text{Bernoulli}(p)$ is defined by
 \be
 H_{\sf B}(p)=-p\log p-(1-p)\log (1-p).
 \ee
Define a random variable $N_1$ based on $X$ such that
\be
\begin{split}
&N_1=\left\{
\begin{array}{l l}
0,& \quad \text{when } X\in\{S_1,\dots, S_L\},\\
1,& \quad  \text{when } X=S_{L+1},
\end{array}\right.
\end{split}
\ee
Since $N_1$ is deterministic given $X$, 
\begin{equation}\label{eqn:chain_RM}
I_{l}(P_X,P_{Y|X})=I(N_1,X;Y)=I(N_1;Y)+I(X;Y|N_1).
\end{equation}
We first find an upper bound on $I(X;Y|N_1)$.
Note that
\begin{equation}
\begin{split}\label{eqn:MI_RM}
&I(X;Y|N_1)\\
&=\left(\sum_{i=1}^L p_i\right) I(X;Y|N_1=0)+p_{L+1}I(X;Y|N_1=1)\\
&=\left(\sum_{i=1}^L p_i\right)\left(H_{\sf B} \left(\sum_{j=1}^L\frac{p_j}{\left(\sum_{i=1}^L p_i\right)}\cdot e^{-|S_j+l|^2}\right)\right.\\
& \left.\qquad\qquad\qquad\quad-\sum_{j=1}^L\frac{p_j}{\left(\sum_{i=1}^L p_i\right)} H_{\sf B}\left( e^{-|S_j+l|^2} \right)\right)\\
\end{split}
\end{equation}
since $ I(X;Y|N_1=1)=0$.
Let $\E_2$ denote the average number of {\it effective} photons used to encode the information in $X$ conditioned on  $N_1=0$:
\begin{equation}\label{eqn:E2}
\E_2=\sum_{j=1}^L\frac{p_j}{\left(\sum_{i=1}^L p_i\right)}\left|S_j-\overline{S}\right|^2
\end{equation}
where $\overline{S}=\sum_{i=1}^L \left({p_i\cdot S_i}\right)/\left({\sum_{i'=1}^L p_{i'}}\right)$ is the average amplitude of the input signal  $\{S_1,\dots, S_L\}$ with normalized probabilities $\{p_1/\left(\sum_{i'=1}^L p_{i'}\right),\dots, p_L/\left(\sum_{i'=1}^L p_{i'}\right)\}$ conditioned on $N_1=0$. When we calculate the average number of {\it effective} photons conditioned on $N_1=0$, we consider the amplitude $|S_i-\overline{S}|$ instead of $S_i$, since we can make a common offset to the signals $\{S_1,\dots, S_L\}$ by using the common control signal $l$ without any cost.
From (\ref{eqn:MI_RM}) and the definition of $R_L(\E)$ in~\eqref{eqn:RLEdef},
\begin{equation}\label{eqn:MI_RM1}
I(X;Y|N_1)\leq \left(\sum_{i=1}^L p_i\right)\cdot R_L(\E_2).
\end{equation}
We next find an upper bound on $I(N_1;Y)$ in~\eqref{eqn:chain_RM}.
Note that the input distribution $P_{N_1}$ is $\{\sum_{i=1}^L p_i,p_{L+1}\}$ and the channel distribution $P_{Y|N_1}$ is
\be
\begin{split}\label{channel11}
&P_{Y|N_1}(Y|N_1=0)\\
&=\left\{
\begin{array}{l l}
\sum_{j=1}^L \frac{p_j}{\left(\sum_{i=1}^L p_i\right)} e^{-|S_j+l|^2}& \quad \text{for }  Y=0,\\
1-\sum_{j=1}^L \frac{p_j}{\left(\sum_{i=1}^L p_i\right)} e^{-|S_j+l|^2}& \quad \text{for }  Y=1,
\end{array}\right.\\
&P_{Y|N_1}(Y|N_1=1)=\left\{
\begin{array}{l l}
e^{-|S_{L+1}+l|^2}& \quad \text{for }  Y=0,\\
1-e^{-|S_{L+1}+l|^2}& \quad \text{for }  Y=1,
\end{array}\right.
\end{split}
\ee
The corresponding mutual information between $N_1$ and $Y$ is
\begin{equation}
\begin{split}
I(N_1;Y)&=H_{\sf B}\left(\sum_{i=1}^{L+1}p_i\cdot e^{-|S_i+l|^2}\right)\\
&\quad-\left(\sum_{i=1}^L p_{i}\right)\cdot H_{\sf B}\left(\sum_{j=1}^L \frac{p_j}{\left(\sum_{i=1}^L p_i\right)} e^{-|S_j+l|^2}\right)\\
&\quad-p_{L+1}\cdot H_{\sf B}\left(e^{-|S_{L+1}+l|^2}\right).
\end{split}
\end{equation}
Define a new channel  distribution $Q_{Y|N_1}$ such that
\be
\begin{split}\label{channel22}
&Q_{Y|N_1}(Y|N_1=0)=\left\{
\begin{array}{l l}
e^{-|\overline{S}+l|^2}& \quad \text{for }  Y=0,\\
1-e^{-|\overline{S}+l|^2}& \quad \text{for }  Y=1,
\end{array}\right.\\
&Q_{Y|N_1}(Y|N_1=1)= P_{Y|N_1}(Y|N_1=1),\; Y\in\{0,1\}.
\end{split}
\ee
For this channel distribution, when $N_1=0$ a coherent state $\ket{\overline{S}}$ is transmitted where $\overline{S}=\sum_{i=1}^L \left({p_i\cdot S_i}\right)/\left({\sum_{i'=1}^L p_{i'}}\right)$, and when $N_1=1$ a coherent state $\ket{S_{L+1}}$ is transmitted. 
The average number $\E_1$ of photons to encode $N_1$ for this new channel equals
\begin{equation}\label{eqn:E1}
\E_1=\left(\sum_{i=1}^L p_i\right)\cdot |\overline{S}|^2+p_{L+1}\cdot |S_{L+1}|^2
\end{equation}
From the definition of $R_L(\E)$ in~\eqref{eqn:RLEdef}, the maximum mutual information between input $N_1$ and output $Y$ with the channel $Q_{Y|N_1}$ is bounded above by
\begin{equation}
I_l(P_{N_1},Q_{Y|N_1})\leq R_2 (\E_1).
\end{equation}
We next show that 
\begin{equation}\label{eqn:comp_p_q}
I_l(P_{N_1},P_{Y|N_1})\leq I_l(P_{N_1},Q_{Y|N_1})
\end{equation}
for any fixed $l$, which will imply that $I(N_1;Y)$ in (\ref{eqn:chain_RM}) is bonded above by
\begin{equation}\label{eqn:IN1Y}
I(N_1;Y)\leq R_2 (\E_1).
\end{equation}
To show~\eqref{eqn:comp_p_q}, we will use the following lemma.
  \begin{lemma}\label{lem:binary_ch_inequal}
  For a binary channel $W_{Y|X}$ with the binary input distribution $P_X$ such that $\{p_0,p_1\}$, let the binary-output channel distribution $W_{Y|X}(Y|X=1)$ be $\{t_1,1-t_1\}$ and $W_{Y|X}(Y|X=1)$  be $\{t_0,1-t_0\}$ for $t_0\geq t_1\geq 0$. Let $f(t_0)$ denote the mutual information $I(P_X,W_{Y|X})$ for a fixed $(t_1,p_0,p_1)$ as a function of $t_0$. Then, $f(t_0)$ decreases monotonically  as $t_0$ decreases and approaches $t_1$.
  \end{lemma}
  \begin{proof}For a fixed $t_1$, let us denote the channel distribution $W_{Y|X}$ as a function of $t_0$ by a matrix $W_{t_0}:=
\left( \begin{array}{cc}
t_0 & 1-t_0  \\
t_1 & 1-t_1 \end{array} \right)$. For $t_2$ such that $t_0\geq t_2\geq t_1$,  there exists $ r \in [0,1)$ such that $r\cdot W_{t_0}+(1-r)\cdot W_{t_1}=W_{t_2}$. Since mutual information $I(P_X,W_{Y|X})$ is convex in $W_{Y|X}$ for a fixed $P_X$, $f(t_0)$ is also convex in $t_0$. Therefore, $r\cdot f(t_0)+(1-r)\cdot f(t_1)\geq f(t_2)$. Since $f(t_1)=0$, the convexity gives $f(t_0)\geq  r\cdot f(t_0)\geq f(t_2)$ for any $(t_0,t_2)$ such that $1\geq t_0\geq t_2\geq t_1\geq 0$. This implies that $f(t_0)$ decreases monotonically as $t_0(>t_1)$ decreases and approaches $t_1$.
  \end{proof}
  
 For $P_{Y|N_1}$ in~\eqref{channel11} and $Q_{Y|N_1}$ in~\eqref{channel22}, if we show 
  \begin{equation}\label{eqn:inequl_3}
 e^{-|S_{L+1}+l|^2} \leq \sum_{j=1}^L \frac{p_j}{\left(\sum_{i=1}^L p_i\right)} e^{-|S_j+l|^2} \leq e^{-|\overline{S}+l|^2},
 \end{equation}
 Lemma \ref{lem:binary_ch_inequal}  implies (\ref{eqn:comp_p_q}).
 In (\ref{eqn:inequl_3}), the first inequality is valid from the ordering of $\{S_1,\dots,S_{L+1}\}$ that satisfies (\ref{eqn:assumption_M+1}). The second inequality is also valid since $e^{-|x+l|^2}$ is concave in $x$ when $|x+l|^2\leq 1/2$, and  $|S_j+l|^2$ for $j=1,\dots, L$ as well as $|\overline{S}+l|^2$, which are the mean photon number received per channel use for each input signal $S_j$ and $\overline{S}$, respectively,  are sufficiently small  due to our assumption of  very short symbol period $\Delta\to 0$.
Therefore,~\eqref{eqn:inequl_3} is valid and by using Lemma \ref{lem:binary_ch_inequal} we can show~\eqref{eqn:comp_p_q}, which implies~\eqref{eqn:IN1Y}. By plugging the upper bounds on $I(N_1;Y)$ in (\ref{eqn:IN1Y}) and on $I(X;Y|N_1)$ in (\ref{eqn:MI_RM1}) into~\eqref{eqn:chain_RM}, we obtain
\begin{equation}\label{eqn:chain_E1_E2}
I_{l}(P_X,P_{Y|X})\leq R_2(\E_1)+ \left(\sum_{i=1}^L p_i\right)\cdot R_L(\E_2)
\end{equation}
where $\E_1$ and $\E_2$ are defined as (\ref{eqn:E1}) and (\ref{eqn:E2}), respectively.
Moreover, it can be shown that 
\begin{equation}
\begin{split}
&\E_1+ \left(\sum_{i=1}^L p_i\right) \cdot \E_2\\
&=\sum_{i=1}^L p_i \left|S_i-\overline{S}\right|^2+\left(\sum_{i=1}^L p_i\right)|\overline{S}|^2+p_{L+1} |S_{L+1}|^2\\
&=\left(\sum_{i=1}^L p_i\right) \left(2|\overline{S}|^2+|S_i|^2-2S_i\overline{S}\right)+p_{L+1} |S_{L+1}|^2\\
&= \sum_{i=1}^{L+1} p_i |S_i|^2=\E.
\end{split}
\end{equation}
When we denote $\E_1=(1-\alpha)  \E$ and $\E_2=\alpha\E/\beta$ for some $\alpha\in(0,1)$ and $\beta:=\left(\sum_{i=1}^L p_i\right)<1$, the upper bound on $I_{l}(P_X,P_{Y|X})$ in (\ref{eqn:chain_E1_E2}) becomes
\begin{equation}\label{eqn:D32}
I_{l}(P_X,P_{Y|X})\leq R_2\left((1-\alpha) \E\right)+\beta R_L\left(\alpha\cdot\E/\beta\right).
\end{equation}
From (\ref{eqn:R2c}) and the assumption (\ref{eqn:RLc}), $I_{l}(P_X,P_{Y|X})$ in the bound~\eqref{eqn:D32} can be further bounded above as
\begin{equation}
\begin{split}\label{eqn:D34}
&I_{l}(P_X,P_{Y|X})\\
&\leq \left((1-\alpha)\E\right)\log\frac{1}{\left((1-\alpha)\E\right)}\\
&\quad-\left((1-\alpha)\E\right)\log\log\frac{1}{\left((1-\alpha)\E\right)}\\
&\quad +\beta \cdot \left(\left(\alpha\cdot\E/\beta\right) \log\frac{1}{\left(\alpha\cdot\E/\beta\right)}\right.\\
&\qquad\qquad\left.-\left(\alpha\cdot\E/\beta\right)\log\log\frac{1}{\left(\alpha\cdot\E/\beta\right)}\right)+O(\E)\\
&\leq\E\log\frac{1}{\E}-\E\log\log\frac{1}{\E}+O(\E)
\end{split}
\end{equation}
for any $0<\alpha,\beta<1$.
This inequality holds for every input set $\mathcal{X}$ of $(L+1)$ real-valued elements with $per_i>0$, for $i=1,\dots, L+1$, under the mean-photon-number constraint of $\E$, regardless of the choice of the control signal $l\in\mathbb{R}$. 

We next extend this result for complex-valued input signals with mean photon number $\E$. Let $\E_{\sf R}$ denote the mean photon number of complex-valued coherent state embedded in real part of the signal, and  $\E_{\sf I}$ be that embedded in imaginary part of the signal. Then, $\E_{\sf R}$ and $\E_{\sf I}$ should satisfy $\E_{\sf R}+\E_{\sf I}=\E$. 
For the optical channel of interest, which is generated by the coherent receiver, when the  input coherent state $\ket{S}$ with complex-field amplitude $S\in\mathbb{C}$ is mixed with a local control signal to generate $\ket{S+l}$ for some $l\in\mathbb{C}$, the resulting channel output follows Poisson process of rate $|S+l|^2=(\mathrm{Re}(S+l))^2+(\mathrm{Im}(S+l))^2$. 
Moreover, this output Poisson process can be decomposed into two independent Poisson processes of rate $(\mathrm{Re}(S+l))^2$ and $(\mathrm{Im}(S+l))^2$, respectively.
Therefore, the capacity of the optical channel with complex-valued coherent states of mean photon number $\E$ is equal to the sum of capacities of two optical channels, whose inputs are real-valued coherent states satisfying the constraints on mean photon numbers, $\E_{\sf R}$ and $\E_{\sf I}$, respectively. By using the upper bound~\eqref{eqn:D34} on the capacity of the optical channel with real-valued arbitrary $(L+1)$ inputs, we can bound the maximum capacity $R_{L+1}(\E)$ with arbitrary $(L+1)$-complex-valued coherent states as
\be
\begin{split}\label{eqn:D35}
R_{L+1}(\E)&\le \E_{\sf R}\log\frac{1}{\E_{\sf R}}-\E_{\sf R}\log\log\frac{1}{\E_{\sf R}}+O(\E_{\sf R})\\
&\quad + \E_{\sf I}\log\frac{1}{\E_{\sf I}}-\E_{\sf I}\log\log\frac{1}{\E_{\sf I}}+O(\E_{\sf I}).
\end{split}
\ee
By using the fact that $\E_{\sf R}+\E_{\sf I}=\E$, we can show that the bound~\eqref{eqn:D35} can be written as
\begin{equation}
R_{L+1}(\E)\leq \E\log\frac{1}{\E}-\E\log\log\frac{1}{\E}+O(\E)
\end{equation}
as $\E\to 0$.
Finally, by mathematical induction,~\eqref{eqn:core_ineq} is true for any input set $\mathcal{X}\subset\mathbb{C}$ with finite cardinality. This completes the proof of Theorem \ref{thm:cap_coh}.

\bibliographystyle{apsrev4-1} 

\begin{thebibliography}{25}%
\makeatletter
\providecommand \@ifxundefined [1]{%
 \@ifx{#1\undefined}
}%
\providecommand \@ifnum [1]{%
 \ifnum #1\expandafter \@firstoftwo
 \else \expandafter \@secondoftwo
 \fi
}%
\providecommand \@ifx [1]{%
 \ifx #1\expandafter \@firstoftwo
 \else \expandafter \@secondoftwo
 \fi
}%
\providecommand \natexlab [1]{#1}%
\providecommand \enquote  [1]{``#1''}%
\providecommand \bibnamefont  [1]{#1}%
\providecommand \bibfnamefont [1]{#1}%
\providecommand \citenamefont [1]{#1}%
\providecommand \href@noop [0]{\@secondoftwo}%
\providecommand \href [0]{\begingroup \@sanitize@url \@href}%
\providecommand \@href[1]{\@@startlink{#1}\@@href}%
\providecommand \@@href[1]{\endgroup#1\@@endlink}%
\providecommand \@sanitize@url [0]{\catcode `\\12\catcode `\$12\catcode
  `\&12\catcode `\#12\catcode `\^12\catcode `\_12\catcode `\%12\relax}%
\providecommand \@@startlink[1]{}%
\providecommand \@@endlink[0]{}%
\providecommand \url  [0]{\begingroup\@sanitize@url \@url }%
\providecommand \@url [1]{\endgroup\@href {#1}{\urlprefix }}%
\providecommand \urlprefix  [0]{URL }%
\providecommand \Eprint [0]{\href }%
\providecommand \doibase [0]{http://dx.doi.org/}%
\providecommand \selectlanguage [0]{\@gobble}%
\providecommand \bibinfo  [0]{\@secondoftwo}%
\providecommand \bibfield  [0]{\@secondoftwo}%
\providecommand \translation [1]{[#1]}%
\providecommand \BibitemOpen [0]{}%
\providecommand \bibitemStop [0]{}%
\providecommand \bibitemNoStop [0]{.\EOS\space}%
\providecommand \EOS [0]{\spacefactor3000\relax}%
\providecommand \BibitemShut  [1]{\csname bibitem#1\endcsname}%
\let\auto@bib@innerbib\@empty
\bibitem [{\citenamefont {Chung}\ \emph
  {et~al.}(2011{\natexlab{a}})\citenamefont {Chung}, \citenamefont {Guha},\
  and\ \citenamefont {Zheng}}]{chung2011capacityISIT}%
  \BibitemOpen
  \bibfield  {author} {\bibinfo {author} {\bibfnamefont {H.~W.}\ \bibnamefont
  {Chung}}, \bibinfo {author} {\bibfnamefont {S.}~\bibnamefont {Guha}}, \ and\
  \bibinfo {author} {\bibfnamefont {L.}~\bibnamefont {Zheng}},\ }in\ \href@noop
  {} {\emph {\bibinfo {booktitle} {2011 IEEE International Symposium on
  Information Theory Proceedings (ISIT)}}}\ (\bibinfo {organization} {IEEE},\
  \bibinfo {year} {2011})\ pp.\ \bibinfo {pages} {284--288}\BibitemShut
  {NoStop}%
\bibitem [{\citenamefont {Chung}\ \emph
  {et~al.}(2011{\natexlab{b}})\citenamefont {Chung}, \citenamefont {Guha},\
  and\ \citenamefont {Zheng}}]{chung2011capacity}%
  \BibitemOpen
  \bibfield  {author} {\bibinfo {author} {\bibfnamefont {H.~W.}\ \bibnamefont
  {Chung}}, \bibinfo {author} {\bibfnamefont {S.}~\bibnamefont {Guha}}, \ and\
  \bibinfo {author} {\bibfnamefont {L.}~\bibnamefont {Zheng}},\ }in\ \href@noop
  {} {\emph {\bibinfo {booktitle} {2011 49th Annual Allerton Conference on
  Communication, Control, and Computing (Allerton)}}}\ (\bibinfo {organization}
  {IEEE},\ \bibinfo {year} {2011})\ pp.\ \bibinfo {pages}
  {879--885}\BibitemShut {NoStop}%
\bibitem [{\citenamefont {Shamai}(1990)}]{Shamai90}%
  \BibitemOpen
  \bibfield  {author} {\bibinfo {author} {\bibfnamefont {S.~S.}\ \bibnamefont
  {Shamai}},\ }\href@noop {} {\bibfield  {journal} {\bibinfo  {journal} {IEE
  Proceedings I (Communications, Speech and Vision)}\ }\textbf {\bibinfo
  {volume} {137}},\ \bibinfo {pages} {424} (\bibinfo {year}
  {1990})}\BibitemShut {NoStop}%
\bibitem [{\citenamefont {Wyner}(1988)}]{Wyner88_1}%
  \BibitemOpen
  \bibfield  {author} {\bibinfo {author} {\bibfnamefont {A.~D.}\ \bibnamefont
  {Wyner}},\ }\href@noop {} {\bibfield  {journal} {\bibinfo  {journal} {IEEE
  Transactions on Information Theory}\ }\textbf {\bibinfo {volume} {34}},\
  \bibinfo {pages} {1462} (\bibinfo {year} {1988})}\BibitemShut {NoStop}%
\bibitem [{\citenamefont {Wang}\ and\ \citenamefont
  {Wornell}(2014)}]{wang2014refined}%
  \BibitemOpen
  \bibfield  {author} {\bibinfo {author} {\bibfnamefont {L.}~\bibnamefont
  {Wang}}\ and\ \bibinfo {author} {\bibfnamefont {G.~W.}\ \bibnamefont
  {Wornell}},\ }\href@noop {} {\bibfield  {journal} {\bibinfo  {journal} {IEEE
  Transactions on Information Theory}\ }\textbf {\bibinfo {volume} {60}},\
  \bibinfo {pages} {4299} (\bibinfo {year} {2014})}\BibitemShut {NoStop}%
\bibitem [{\citenamefont {Dolinar}(1973)}]{DOL73}%
  \BibitemOpen
  \bibfield  {author} {\bibinfo {author} {\bibfnamefont {S.~J.}\ \bibnamefont
  {Dolinar}},\ }\href@noop {} {\bibfield  {journal} {\bibinfo  {journal} {MIT
  Research Laboratory of Electronics Quarterly Progress Report}\ }\textbf
  {\bibinfo {volume} {111}},\ \bibinfo {pages} {115} (\bibinfo {year}
  {1973})}\BibitemShut {NoStop}%
\bibitem [{\citenamefont {Yuen}\ \emph {et~al.}(1975)\citenamefont {Yuen},
  \citenamefont {Kennedy},\ and\ \citenamefont {Lax}}]{YKL1975}%
  \BibitemOpen
  \bibfield  {author} {\bibinfo {author} {\bibfnamefont {H.~P.}\ \bibnamefont
  {Yuen}}, \bibinfo {author} {\bibfnamefont {R.~S.}\ \bibnamefont {Kennedy}}, \
  and\ \bibinfo {author} {\bibfnamefont {M.}~\bibnamefont {Lax}},\ }\href@noop
  {} {\bibfield  {journal} {\bibinfo  {journal} {IEEE Transactions on
  Information Theory}\ }\textbf {\bibinfo {volume} {IT-21}},\ \bibinfo {pages}
  {125134} (\bibinfo {year} {1975})}\BibitemShut {NoStop}%
\bibitem [{\citenamefont {Helstrom}\ \emph {et~al.}(1976)\citenamefont
  {Helstrom} \emph {et~al.}}]{Hel76}%
  \BibitemOpen
  \bibfield  {author} {\bibinfo {author} {\bibfnamefont {C.~W.}\ \bibnamefont
  {Helstrom}} \emph {et~al.},\ }\href@noop {} {\emph {\bibinfo {title} {Quantum
  detection and estimation theory}}},\ Vol.~\bibinfo {volume} {84}\ (\bibinfo
  {publisher} {Academic press New York},\ \bibinfo {year} {1976})\BibitemShut
  {NoStop}%
\bibitem [{\citenamefont {da~Silva}\ \emph {et~al.}(2013)\citenamefont
  {da~Silva}, \citenamefont {Guha},\ and\ \citenamefont {Dutton}}]{Silva2013}%
  \BibitemOpen
  \bibfield  {author} {\bibinfo {author} {\bibfnamefont {M.~P.}\ \bibnamefont
  {da~Silva}}, \bibinfo {author} {\bibfnamefont {S.}~\bibnamefont {Guha}}, \
  and\ \bibinfo {author} {\bibfnamefont {Z.}~\bibnamefont {Dutton}},\
  }\href@noop {} {\bibfield  {journal} {\bibinfo  {journal} {Physical Review
  A}\ }\textbf {\bibinfo {volume} {87}},\ \bibinfo {pages} {052320} (\bibinfo
  {year} {2013})}\BibitemShut {NoStop}%
\bibitem [{\citenamefont {Holevo}(1998{\natexlab{a}})}]{holevo1998quantum}%
  \BibitemOpen
  \bibfield  {author} {\bibinfo {author} {\bibfnamefont {A.~S.}\ \bibnamefont
  {Holevo}},\ }\href@noop {} {\bibfield  {journal} {\bibinfo  {journal}
  {Russian Mathematical Surveys}\ }\textbf {\bibinfo {volume} {53}},\ \bibinfo
  {pages} {1295} (\bibinfo {year} {1998}{\natexlab{a}})}\BibitemShut {NoStop}%
\bibitem [{Note1()}]{Note1}%
  \BibitemOpen
  \bibinfo {note} {One has to be careful in using the binary-output channel as
  an approximation of the Poisson channel. As we are optimizing over the
  control signal, it is not obvious that the resulting $\lambda _i$'s are
  bounded. In other words, the mean of the Poisson distributions, $\lambda _i
  \Delta $, might not be small. The assumption of either $0$ or $1$ arrival,
  and the approximation in the corresponding probabilities, can be justified as
  follows. First, a single photon detector is much more practical, given the
  current state of technology, that a fully number-resolving high bandwidth
  photon counter. A single photon detector can sense whether or not any number
  of photons arrives during a time interval $\Delta $, but cannot count the
  number of photon arrivals, especially as $\Delta \to 0$. So, the
  binary-output channel model is much more practical than the Poisson-output
  channel model. Second, when we want to maximize the ability to distinguish
  between two hypotheses $H=0, 1$, we essentially need to distinguish between
  the signal amplitudes $S_0$ and $S_1$ using photon arrival events. Adding a
  feedback control signal $l\to \infty $ does not help in distinguishing $S_0$
  and $S_1$. In this sense, we can reason that the optimal $l$ should not make
  $\lambda _i$ unbounded.}\BibitemShut {Stop}%
\bibitem [{\citenamefont {Takeoka}(2007)}]{Takeoka2007}%
  \BibitemOpen
  \bibfield  {author} {\bibinfo {author} {\bibfnamefont {M.}~\bibnamefont
  {Takeoka}},\ }\href@noop {} {\bibfield  {journal} {\bibinfo  {journal}
  {Optics and Spectroscopy}\ }\textbf {\bibinfo {volume} {103}},\ \bibinfo
  {pages} {98} (\bibinfo {year} {2007})}\BibitemShut {NoStop}%
\bibitem [{\citenamefont {Walgate}\ \emph {et~al.}(2000)\citenamefont
  {Walgate}, \citenamefont {Short}, \citenamefont {Hardy},\ and\ \citenamefont
  {Vedral}}]{Walgate2000}%
  \BibitemOpen
  \bibfield  {author} {\bibinfo {author} {\bibfnamefont {J.}~\bibnamefont
  {Walgate}}, \bibinfo {author} {\bibfnamefont {A.~J.}\ \bibnamefont {Short}},
  \bibinfo {author} {\bibfnamefont {L.}~\bibnamefont {Hardy}}, \ and\ \bibinfo
  {author} {\bibfnamefont {V.}~\bibnamefont {Vedral}},\ }\href@noop {}
  {\bibfield  {journal} {\bibinfo  {journal} {Physical Review Letters}\
  }\textbf {\bibinfo {volume} {85}},\ \bibinfo {pages} {4972} (\bibinfo {year}
  {2000})}\BibitemShut {NoStop}%
\bibitem [{\citenamefont {Holevo}(1998{\natexlab{b}})}]{Holevo98}%
  \BibitemOpen
  \bibfield  {author} {\bibinfo {author} {\bibfnamefont {A.~S.}\ \bibnamefont
  {Holevo}},\ }\href {\doibase 10.1109/18.651037} {\bibfield  {journal}
  {\bibinfo  {journal} {IEEE Transactions on Information Theory}\ }\textbf
  {\bibinfo {volume} {44}},\ \bibinfo {pages} {269} (\bibinfo {year}
  {1998}{\natexlab{b}})}\BibitemShut {NoStop}%
\bibitem [{\citenamefont {Schumacher}\ and\ \citenamefont
  {Westmoreland}(1997)}]{Schumacher97}%
  \BibitemOpen
  \bibfield  {author} {\bibinfo {author} {\bibfnamefont {B.}~\bibnamefont
  {Schumacher}}\ and\ \bibinfo {author} {\bibfnamefont {M.~D.}\ \bibnamefont
  {Westmoreland}},\ }\href@noop {} {\bibfield  {journal} {\bibinfo  {journal}
  {Physical Review A}\ }\textbf {\bibinfo {volume} {56}},\ \bibinfo {pages}
  {131} (\bibinfo {year} {1997})}\BibitemShut {NoStop}%
\bibitem [{\citenamefont {Giovannetti}\ \emph {et~al.}(2004)\citenamefont
  {Giovannetti}, \citenamefont {Guha}, \citenamefont {Lloyd}, \citenamefont
  {Maccone}, \citenamefont {Shapiro},\ and\ \citenamefont
  {Yuen}}]{giovannetti2004classical}%
  \BibitemOpen
  \bibfield  {author} {\bibinfo {author} {\bibfnamefont {V.}~\bibnamefont
  {Giovannetti}}, \bibinfo {author} {\bibfnamefont {S.}~\bibnamefont {Guha}},
  \bibinfo {author} {\bibfnamefont {S.}~\bibnamefont {Lloyd}}, \bibinfo
  {author} {\bibfnamefont {L.}~\bibnamefont {Maccone}}, \bibinfo {author}
  {\bibfnamefont {J.~H.}\ \bibnamefont {Shapiro}}, \ and\ \bibinfo {author}
  {\bibfnamefont {H.~P.}\ \bibnamefont {Yuen}},\ }\href@noop {} {\bibfield
  {journal} {\bibinfo  {journal} {Physical Review Letters}\ }\textbf {\bibinfo
  {volume} {92}},\ \bibinfo {pages} {027902} (\bibinfo {year}
  {2004})}\BibitemShut {NoStop}%
\bibitem [{\citenamefont {Lapidoth}\ \emph {et~al.}(2011)\citenamefont
  {Lapidoth}, \citenamefont {Shapiro}, \citenamefont {Venkatesan},\ and\
  \citenamefont {Wang}}]{Ligong08}%
  \BibitemOpen
  \bibfield  {author} {\bibinfo {author} {\bibfnamefont {A.}~\bibnamefont
  {Lapidoth}}, \bibinfo {author} {\bibfnamefont {J.~H.}\ \bibnamefont
  {Shapiro}}, \bibinfo {author} {\bibfnamefont {V.}~\bibnamefont {Venkatesan}},
  \ and\ \bibinfo {author} {\bibfnamefont {L.}~\bibnamefont {Wang}},\
  }\href@noop {} {\bibfield  {journal} {\bibinfo  {journal} {IEEE Transactions
  on Information Theory}\ }\textbf {\bibinfo {volume} {57}},\ \bibinfo {pages}
  {3260} (\bibinfo {year} {2011})}\BibitemShut {NoStop}%
\bibitem [{\citenamefont {Verd{\'u}}(2002)}]{Verdu02}%
  \BibitemOpen
  \bibfield  {author} {\bibinfo {author} {\bibfnamefont {S.}~\bibnamefont
  {Verd{\'u}}},\ }\href@noop {} {\bibfield  {journal} {\bibinfo  {journal}
  {IEEE Transactions on Information Theory}\ }\textbf {\bibinfo {volume}
  {48}},\ \bibinfo {pages} {1319} (\bibinfo {year} {2002})}\BibitemShut
  {NoStop}%
\bibitem [{\citenamefont {Zheng}\ \emph {et~al.}(2007)\citenamefont {Zheng},
  \citenamefont {Tse},\ and\ \citenamefont {M{\'e}dard}}]{ZTM03}%
  \BibitemOpen
  \bibfield  {author} {\bibinfo {author} {\bibfnamefont {L.}~\bibnamefont
  {Zheng}}, \bibinfo {author} {\bibfnamefont {D.~N.}\ \bibnamefont {Tse}}, \
  and\ \bibinfo {author} {\bibfnamefont {M.}~\bibnamefont {M{\'e}dard}},\
  }\href@noop {} {\bibfield  {journal} {\bibinfo  {journal} {IEEE Transactions
  on Information Theory}\ }\textbf {\bibinfo {volume} {53}},\ \bibinfo {pages}
  {976} (\bibinfo {year} {2007})}\BibitemShut {NoStop}%
\bibitem [{\citenamefont {Reck}\ \emph {et~al.}(1994)\citenamefont {Reck},
  \citenamefont {Zeilinger}, \citenamefont {Bernstein},\ and\ \citenamefont
  {Bertani}}]{Reck1994}%
  \BibitemOpen
  \bibfield  {author} {\bibinfo {author} {\bibfnamefont {M.}~\bibnamefont
  {Reck}}, \bibinfo {author} {\bibfnamefont {A.}~\bibnamefont {Zeilinger}},
  \bibinfo {author} {\bibfnamefont {H.~J.}\ \bibnamefont {Bernstein}}, \ and\
  \bibinfo {author} {\bibfnamefont {P.}~\bibnamefont {Bertani}},\ }\href@noop
  {} {\bibfield  {journal} {\bibinfo  {journal} {Physical Review Letters}\
  }\textbf {\bibinfo {volume} {73}},\ \bibinfo {pages} {58} (\bibinfo {year}
  {1994})}\BibitemShut {NoStop}%
\bibitem [{\citenamefont {Guha}(2011)}]{Guh10}%
  \BibitemOpen
  \bibfield  {author} {\bibinfo {author} {\bibfnamefont {S.}~\bibnamefont
  {Guha}},\ }\href@noop {} {\bibfield  {journal} {\bibinfo  {journal} {Physical
  Review Letters}\ }\textbf {\bibinfo {volume} {106}},\ \bibinfo {pages}
  {240502} (\bibinfo {year} {2011})}\BibitemShut {NoStop}%
\bibitem [{\citenamefont {Rosati}\ \emph {et~al.}(2017)\citenamefont {Rosati},
  \citenamefont {Mari},\ and\ \citenamefont {Giovannetti}}]{rosati2017}%
  \BibitemOpen
  \bibfield  {author} {\bibinfo {author} {\bibfnamefont {M.}~\bibnamefont
  {Rosati}}, \bibinfo {author} {\bibfnamefont {A.}~\bibnamefont {Mari}}, \ and\
  \bibinfo {author} {\bibfnamefont {V.}~\bibnamefont {Giovannetti}},\
  }\href@noop {} {\bibfield  {journal} {\bibinfo  {journal} {arXiv preprint
  arXiv:1703.05701}\ } (\bibinfo {year} {2017})}\BibitemShut {NoStop}%
\bibitem [{\citenamefont {Sohma}\ and\ \citenamefont
  {Hirota}(2000)}]{sohma2000binary}%
  \BibitemOpen
  \bibfield  {author} {\bibinfo {author} {\bibfnamefont {M.}~\bibnamefont
  {Sohma}}\ and\ \bibinfo {author} {\bibfnamefont {O.}~\bibnamefont {Hirota}},\
  }\href@noop {} {\bibfield  {journal} {\bibinfo  {journal} {Physical Review
  A}\ }\textbf {\bibinfo {volume} {62}},\ \bibinfo {pages} {052312} (\bibinfo
  {year} {2000})}\BibitemShut {NoStop}%
\bibitem [{\citenamefont {Cover}\ and\ \citenamefont
  {Thomas}(2012)}]{cover2012elements}%
  \BibitemOpen
  \bibfield  {author} {\bibinfo {author} {\bibfnamefont {T.~M.}\ \bibnamefont
  {Cover}}\ and\ \bibinfo {author} {\bibfnamefont {J.~A.}\ \bibnamefont
  {Thomas}},\ }\href@noop {} {\emph {\bibinfo {title} {Elements of information
  theory}}}\ (\bibinfo  {publisher} {John Wiley \& Sons},\ \bibinfo {year}
  {2012})\BibitemShut {NoStop}%
\bibitem [{\citenamefont {Chung}\ \emph {et~al.}(2016)\citenamefont {Chung},
  \citenamefont {Guha},\ and\ \citenamefont {Zheng}}]{chung2016}%
  \BibitemOpen
  \bibfield  {author} {\bibinfo {author} {\bibfnamefont {H.~W.}\ \bibnamefont
  {Chung}}, \bibinfo {author} {\bibfnamefont {S.}~\bibnamefont {Guha}}, \ and\
  \bibinfo {author} {\bibfnamefont {L.}~\bibnamefont {Zheng}},\ }\href@noop {}
  {\bibfield  {journal} {\bibinfo  {journal} {IEEE Transactions on Information
  Theory}\ }\textbf {\bibinfo {volume} {62}},\ \bibinfo {pages} {5938}
  (\bibinfo {year} {2016})}\BibitemShut {NoStop}%
\end{thebibliography}
%

\end{document}